\documentclass[a4paper,10pt,twocolumn]{article}
\pdfoutput=1
\usepackage[utf8]{inputenc}
\usepackage{amsmath,amsthm}

\usepackage{stmaryrd} 
\usepackage{dsfont} 
\usepackage{graphicx}
\usepackage{enumerate}
\usepackage{subfigure}

\numberwithin{equation}{section} 
\newtheorem{theorem}{Theorem}[]
\newtheorem{lemma}{Lemma}[]
\newtheorem{corollary}{Corollary}[]
\theoremstyle{definition}
\newtheorem{definition}{Definition}[]
\theoremstyle{definition}
\newtheorem{remark}{Remark}[]

\newcommand{\ddt}[1]{\frac{d #1}{dt}}
\newcommand{\pddt}[1]{\frac{\partial #1}{\partial t}}
\newcommand{\DDT}[1]{\frac{D #1}{Dt}}
\newcommand{\DDTp}[1]{\frac{D^+ #1}{Dt}}
\newcommand{\DDTm}[1]{\frac{D^- #1}{Dt}}
\newcommand{\DDTpm}[1]{\frac{D^\pm #1}{Dt}}
\newcommand{\DSigmaDT}[1]{\frac{D^\Sigma #1}{Dt}}
\newcommand{\DT}[1]{D_t #1}
\newcommand{\thomas}{\partial^\Sigma_t}

\DeclareMathOperator{\dist}{dist}

\newcommand{\inproduct}[2]{\left\langle #1 , #2 \right\rangle}
\newcommand{\jump}[1]{\left\llbracket #1 \right\rrbracket}
\DeclareMathOperator{\regular}{\mathcal{J}}

\newcommand{\divergence}[1]{\nabla \cdot #1}
\DeclareMathOperator{\divsigma}{div_\Sigma}
\DeclareMathOperator{\nablasigma}{\nabla_\Sigma}
\DeclareMathOperator{\nablagamma}{\nabla_\Gamma}
\newcommand{\transpose}{{\sf T}}
\newcommand{\norm}[1]{\lVert #1 \rVert} 
\DeclareMathOperator{\gr}{gr} 

\newcommand{\RR}{\mathds{R}} 
\newcommand{\NN}{\mathds{N}} 
\DeclareMathOperator{\energy}{\mathcal{E}} 
\DeclareMathOperator{\wet}{W} 
\newcommand{\visc}{\eta} 
\newcommand{\il}{a} 
\newcommand{\ftheta}{f} 
\newcommand{\gtheta}{g} 
\DeclareMathOperator{\sigmawet}{\sigma_w}

\DeclareMathOperator{\thetaeq}{\theta_{eq}}

\newcommand{\move}{\mathcal{M}} 
\newcommand{\velspace}{\mathcal{V}} 
\newcommand{\flowmap}[1]{\Phi_{#1}}
\newcommand{\tangentspace}[1]{T_{#1}}

\newcommand{\nsigma}{n_\Sigma}

\newcommand{\ndomega}{n_{\partial\Omega}}
\newcommand{\ngamma}{n_\Gamma}
\newcommand{\tgamma}{t_\Gamma}

\newcommand{\psigma}{\mathcal{P}_\Sigma}
\newcommand{\pdomega}{\mathcal{P}_{\partial\Omega}}

\newcommand{\vsigma}{v_\Sigma} 
\newcommand{\clspeed}{V_\Gamma} 
\newcommand{\normalspeed}{V_\Sigma} 
\DeclareMathOperator{\vwall}{v_{w}}
\usepackage{authblk} 
\usepackage{hyperref}
\usepackage{microtype}
\usepackage[a4paper,left=2cm,right=2cm,top=3.2cm,bottom=3.2cm]{geometry}
\date{}

\usepackage{chngcntr}
\usepackage{apptools}
\AtAppendix{\counterwithin{theorem}{section}}
\AtAppendix{\counterwithin{lemma}{section}} 
\AtAppendix{\counterwithin{corollary}{section}}
\AtAppendix{\counterwithin{remark}{section}}
\AtAppendix{\counterwithin{definition}{section}} 

\title{A Kinematic Evolution Equation for the\\ Dynamic Contact Angle and some Consequences}
\author[1]{Mathis Fricke}
\author[2]{Matthias Köhne} 
\author[1,3]{Dieter Bothe}

\affil[1]{Department of Mathematics, TU Darmstadt}
\affil[2]{Department of Mathematics, HHU Düsseldorf}
\affil[3]{Profile Area Thermo-Fluids \& Interfaces, TU Darmstadt}

\begin{document}

\twocolumn[\begin{@twocolumnfalse}

\vspace{-1.5cm}
\maketitle
\thispagestyle{empty}

\vspace{-1cm}
\begin{abstract}
We investigate the moving contact line problem for two-phase incompressible flows with a kinematic approach. The key idea is to derive an evolution equation for the contact angle in terms of the transporting velocity field. It turns out that the resulting equation has a simple structure and expresses the time derivative of the contact angle in terms of the velocity gradient at the solid wall. Together with the additionally imposed boundary conditions for the velocity, it yields a more specific form of the contact angle evolution. Thus, the kinematic evolution equation is a tool to analyze the evolution of the contact angle. Since the transporting velocity field is required only \emph{on} the moving interface, the kinematic evolution equation also applies when the interface moves with its own velocity independent of the fluid velocity.\newline
We apply the developed tool to a class of moving contact line models which employ the Navier slip boundary condition. We derive an explicit form of the contact angle evolution for sufficiently regular solutions, showing that such solutions are unphysical. Within the simplest model, this rigorously shows that the contact angle can only relax to equilibrium if some kind of singularity is present at the contact line.\newline
Moreover, we analyze more general models including surface tension gradients at the contact line, slip at the fluid-fluid interface and mass transfer across the fluid-fluid interface.
\end{abstract}

This preprint was submitted and accepted for publication in \emph{Physica D: Nonlinear Phenomena}.\newline When citing this work, please refer to the journal article: \textbf{DOI:} \href{http://dx.doi.org/10.1016/j.physd.2019.01.008}{10.1016/j.physd.2019.01.008}.\newline

\tableofcontents

\end{@twocolumnfalse}]
\clearpage

\section{Introduction}
We are interested in the evolution of a contact line on a flat and homogeneous solid surface. For this purpose, we consider a two-phase system consisting of two immiscible Newtonian fluids described by the Navier-Stokes equations in the case where the fluid-fluid interface $\Sigma$ has contact with a solid. Employing the sharp interface modeling approach, we describe the interfacial layer as a mathematical surface of zero thickness. However, this surface carries additional physical properties like surface tension described by the surface tension coefficient $\sigma$. The curve of intersection of the fluid-fluid interface with the solid boundary is called the \emph{contact line}. If the length scale of the flow is sufficiently small, the physical processes at the contact line can have a significant influence on the macroscopic behavior of the system. A typical example is the rise of liquid in a capillary tube with a diameter comparable to the capillary length
\[ l_c = \sqrt{\sigma/\rho g}.  \]
The rise height $H$ of the liquid in the equilibrium state can be found by means of energy considerations and is strongly dependent on the equilibrium contact angle $\thetaeq$, i.e.\ the angle of intersection between the fluid-fluid interface and the solid wall in equilibrium. The height given by equation \eqref{eqn:jurins_height} is well-known in the literature as ``Jurin's height''
\begin{align} 
\label{eqn:jurins_height}
H = \frac{2 \sigma \cos \thetaeq}{\rho g R},
\end{align}
see for example \cite{Gennes.2010}.\newline
\newline
It is known since the 19\textsuperscript{th} century that the \emph{equilibrium} contact angle is described by Young's equation \cite{Young.1805} (see Figure~\ref{fig:young}),
\begin{align}
\label{eqn:young}
\sigma \cos \thetaeq + \sigmawet = 0,
\end{align}
where $\sigmawet := \sigma_1 - \sigma_2$ is the specific energy of the wetted surface (relative to the ``dry surface'').

\begin{figure}[bht]
\centering
\includegraphics[width=5cm]{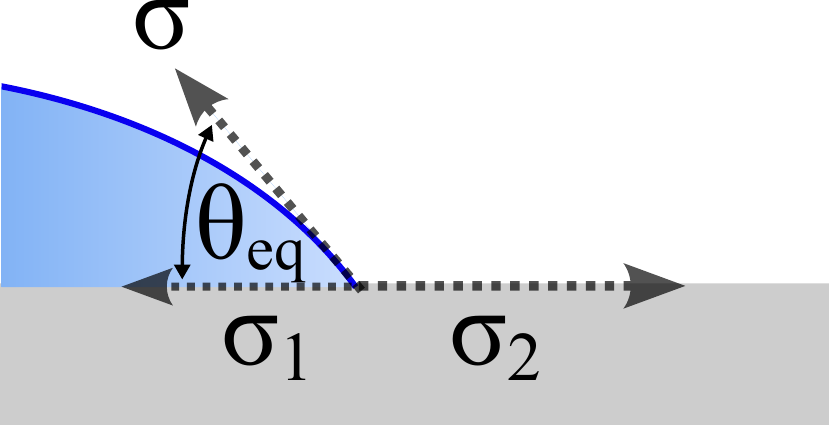}
\caption{Young diagram and equilibrium contact angle.}
\label{fig:young}
\end{figure}

However, the mathematical modeling of the dynamics of wetting turned out to be a challenging problem, leading to a lot of scientific debate and a great variety of different models ranging from molecular to continuum mechanical descriptions, using both sharp and diffuse interface models. For a recent survey on the field see, for example, \cite{Gennes.2010}, \cite{Blake.2006}, \cite{Bonn.2009}, \cite{Shikhmurzaev.2008}, \cite{Snoeijer.2013} and references therein. While a description on the molecular level can be expected to be more accurate in capturing the physics at the contact line, it is limited in terms of length and time scales for practical applications in the natural and engineering sciences as well as in industry. Therefore, it is desirable to have a \emph{continuum} mechanical model with an \emph{effective} description of the necessary physics which is relevant on smaller length scales. However, it turned out that such a model is not straightforward to find. In 1971, Huh and Scriven \cite{Huh.1971} showed that the usual \emph{no-slip} condition for the velocity at solid boundaries is not appropriate for the modeling of moving contact lines (see also \cite{Pukhnachev.1982},\cite{Solonnikov.1997}). Depending on the model and the solution concept, this means that the no-slip condition either leads to the non-existence of solutions with moving contact line or to an infinite viscous dissipation rate. Since then, many attempts have been made to solve this problem, for instance by means of numerical discretizations\footnote{A typical approach to circumvent the problem numerically is to use so-called \emph{numerical slip}. The main observation is (see~\cite{Renardy.2001}) that in many cases an artificial slip is introduced by the discretization itself. Due to this numerical effect the contact line is able to move even though the no-slip condition is used. However, the numerical slip is typically strongly dependent on the grid size. Moreover, the model which is supposed to describe the \emph{physics} of dynamic wetting is in this case purely numerical and has no physical justification.} (for an overview over numerical methods see \cite{Sui.2014}) and/or by replacing the boundary conditions in the model. Essentially, some mechanism is introduced which allows for a tangential slip of the \emph{interfacial} velocity at the solid wall. The most common choice for the velocity boundary condition is the Navier slip condition, already proposed by Navier in the 19\textsuperscript{th} century, which relates the tangential slip to the tangential component of normal stress at the boundary. The boundary condition for the contact angle is frequently motivated by experimental observations. The experimentally measured contact angle, which is always subject to a finite measurement resolution, typically shows a strong correlation to the capillary number,
\[ Ca = \frac{\visc \clspeed}{\sigma}, \]
where $\visc$ and $\clspeed$ denote the dynamic viscosity and the contact line velocity, respectively. Motivated by this observation, many models \emph{prescribe} the contact angle as some function of the capillary number and the equilibrium contact angle, i.e.
\begin{align*} 
\theta = \ftheta(\thetaeq, Ca). 
\end{align*}
It can be shown by means of asymptotic analysis for the \emph{stationary} Stokes equations that the Navier slip condition with a finite slip length makes the viscous dissipation rate finite, while the pressure is still logarithmically singular at the moving contact line (see, e.g., \cite{Huh1977}, \cite{Shikhmurzaev.2006}). This integrable type of singularity is commonly referred to as a \emph{weak singularity}. It is an interesting question, under which circumstances even the weak singularity is removed from the description. In the publications \cite{Ren.2007} and \cite{E.2011}, Ren and Weinan E formulate the expectation that the weak singularity is removed if, instead of a fixed contact angle, a certain model for the dynamic contact angle is applied. The present work shows that this is not the case. Instead, it is shown that the dynamic behavior for sufficiently regular solutions to the simplest model is \emph{unphysical} (if the slip length is finite).\\

The present work tries to contribute to the mathematical modeling of moving contact lines by analyzing the mathematical properties of the discussed models, in particular the combination of boundary conditions at the contact line, by means of a \emph{kinematic} approach. The key idea is to understand how the flow field transports the contact angle. Note that here we consider the most simple case of a flat, perfectly clean solid wall and ideal Newtonian fluids, a situation never met in a real-world experiment. A real surface always has some geometrical and chemical structure leading to additional effects like contact angle hysteresis and pinning. Moreover, it might be interesting to consider more complex liquids and substrates to enhance certain properties for applications. However, it seems meaningful to first study the mathematics of the problem in the simplest setting.

\paragraph{Organization of the paper:} The remainder of this paper is organized as follows. After introducing some basic notation, we recall the ``standard model for two-phase flow'' and its extension to contact lines using the Navier boundary condition in Section~\ref{section:two-phase-flow-model}. An evolution equation for the contact angle, which we refer to as the \emph{kinematic evolution equation}, is derived in Section~\ref{section:geometrical_evolution_equation}. With the help of this result, the time derivative of the contact angle can be expressed by means of the velocity gradient at the contact line.  The application of the kinematic evolution equation to this class of models is discussed in Section~\ref{section:ca_evolution_standard_model}. In particular, it is shown that potential regular solutions to the standard model are unphysical. The contact angle evolution for more general models is briefly discussed in Section~\ref{section:generalizations}. 

\section{Mathematical modeling}
\label{section:two-phase-flow-model}

\subsection[Notation and mathematical setting]{Notation and mathematical\\ setting}
For simplicity, let us assume for this paper that $\Omega$ is a half space, such that the outer normal field $\ndomega$ is constant. This is not a real restriction for the theory since we are only interested in \emph{local} properties. While this assumption simplifies the calculations, the results may be generalized to the case of a curved solid wall.\\

The following definition of a $\mathcal{C}^{1,2}$-family of moving hypersurfaces can also be found in \cite{Kimura.2008}, \cite{Pruss.2016} and in a similar form in \cite{Giga.2006}.
\begin{definition}
Let $I=(a,b)$ be an open interval. A family $\{\Sigma(t)\}_{t \in I}$ with $\Sigma(t) \subset \RR^3$ is called a \emph{$\mathcal{C}^{1,2}$-family of moving hypersurfaces} if the following holds.
\begin{enumerate}[(i)]
 \item Each $\Sigma(t)$ is an orientable $\mathcal{C}^2$-hypersurface in $\RR^3$ with unit normal field denoted as $\nsigma(t,\cdot)$.
 \item The graph of $\Sigma$, given as
 \begin{align}
 \label{eqn:def_moving_interface}
 \move := \gr \Sigma = \bigcup_{t \in I} \{t\} \times \Sigma(t) \subset \RR\times\RR^{3},
 \end{align}
 is a $\mathcal{C}^1$-hypersurface in $\RR \times \RR^3$.
 \item The unit normal field is continuously differentiable on $\move$, i.e.
 \[ \nsigma \in \mathcal{C}^1(\move). \]
\end{enumerate}
A family $\{\overline{\Sigma}(t)\}_{t \in I}$ is called a \emph{$\mathcal{C}^{1,2}$-family of moving hypersurfaces with boundary $\partial\Sigma(t)$} if the following holds.
\begin{enumerate}[(i)]
 \item Each $\overline{\Sigma}(t)$ is an orientable $\mathcal{C}^2$-hypersurface in $\RR^3$ with interior $\Sigma(t)$ and non-empty boundary $\partial\Sigma(t)$, where the unit normal field is denoted by $\nsigma(t,\cdot)$.
 \item The graph of $\overline{\Sigma}$, i.e.
 \[ \gr \overline{\Sigma} = \bigcup_{t \in I} \{t\} \times \overline{\Sigma}(t) \subset \RR\times\RR^{3}, \]
 is a $\mathcal{C}^1$-hypersurface with boundary $\gr(\partial\Sigma)$ in $\RR\times\RR^3$.
  \item The unit normal field is continuously differentiable on $\gr \overline{\Sigma}$, i.e.
 \[ \nsigma \in \mathcal{C}^1(\gr \overline{\Sigma}). \]
\end{enumerate}
\end{definition}
Note that, being the boundary of a submanifold with boundary, the set $\gr(\partial\Sigma)$ is itself a submanifold (without boundary).\newline
\newline
In the remainder of this paper, we consider the following geometrical situation: Let $\Omega \subset \RR^3$ be a half space and let the ``fluid-fluid interface'' $\{\Sigma(t)\}_{t \in I}$ be a $\mathcal{C}^{1,2}$-family of moving hypersurfaces with boundary $\partial\Sigma$ such that
\[ \Sigma(t) \subset \Omega, \ \partial\Sigma(t) \subset \partial\Omega \quad \forall t \in I, \]
i.e.\ the boundary of $\Sigma$ is contained in the domain boundary. The moving fluid-fluid interface decomposes $\Omega$ into two bulk-phases, i.e
\begin{align*}
\Omega = \Omega^+(t) \cup \Omega^-(t) \cup \Sigma(t),
\end{align*}
where the unit normal field $\nsigma$ is pointing from $\Omega^-(t)$ to $\Omega^+(t)$. The \emph{contact line} $\Gamma(t) \subset \partial\Omega$ is the subset of the solid boundary which is in contact with the interface $\Sigma(t)$, i.e.
 \begin{align*}
 \Gamma(t) := \partial\Sigma(t) = \partial\Omega \, \cap \, \overline{\Omega^+}(t) \, \cap \, \overline{\Omega^-}(t) \neq \emptyset.
 \end{align*}
We assume that $\Gamma(t)$ is non-empty and therefore do not consider the process of formation or disappearance of the contact line as a whole. Given a point $x \in \Gamma(t)$, the \emph{contact angle} $\theta$ is defined by the relation
\begin{align}
\label{eqn:theta_definition}
\cos \theta(t,x) := -\inproduct{\nsigma(t,x)}{\ndomega(t,x)}.
\end{align}

\begin{figure}[hbt]
 \centering
 \includegraphics[width=6cm]{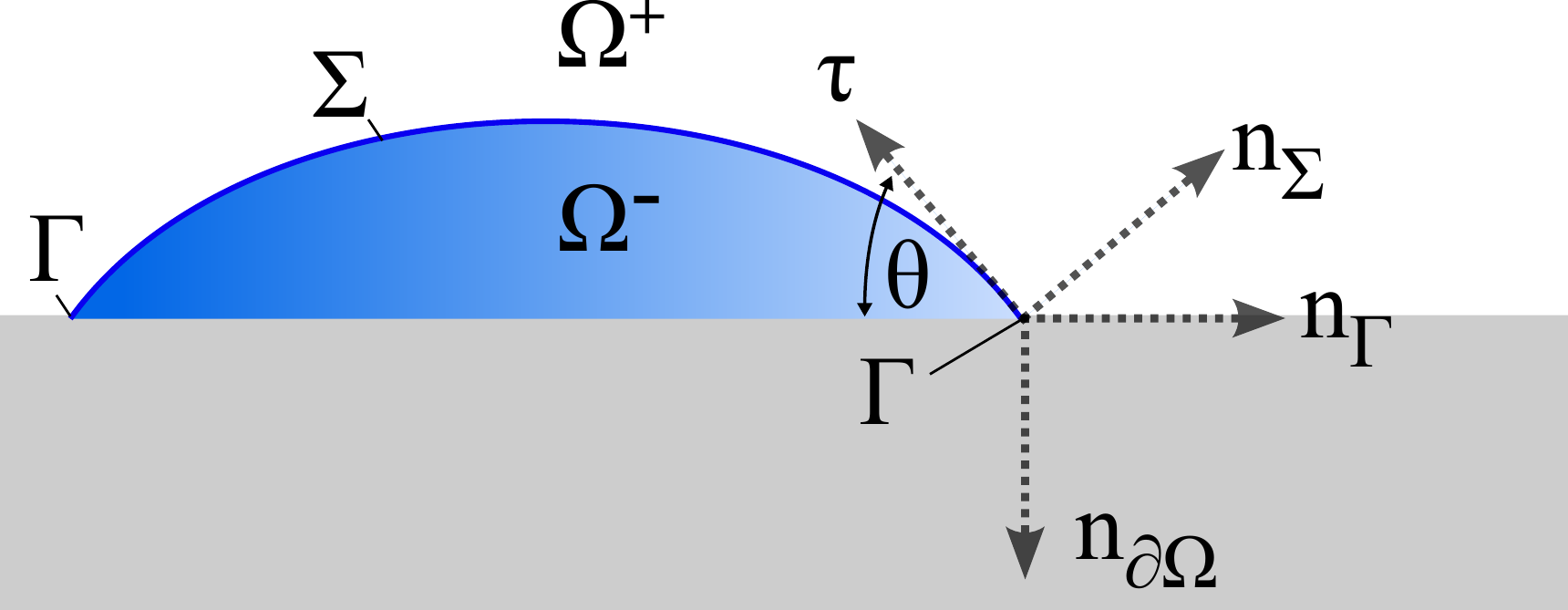}
 \caption{Notation, local coordinate system.}
 \label{fig:nsigma}
\end{figure}

\paragraph{Local coordinate system:} For simplicity of notation, we choose the reference frame where the wall is at rest. Given a point $x \in \Gamma(t)$ at the contact line, we set up a local coordinate system to describe the evolution of the system. A possible choice is to use $\nsigma$ and $\ndomega$ together with a third linear independent direction. However, the vectors $\nsigma$ and $\ndomega$ are, in general, not orthogonal and it is more convenient to introduce a contact line normal vector.

\begin{definition}
For $0 < \theta < \pi$ the \emph{contact line normal vector} $\ngamma$ is defined via projection\footnote{The orthogonal projection operator onto $\partial\Omega$ is given as $\pdomega := \mathds{1} - \inproduct{\ndomega}{\cdot} \ndomega$.} (see Figure~\ref{fig:nsigma}) as
\begin{align}
\label{eqn:def_ngamma}
\ngamma = \frac{\pdomega \, \nsigma}{\lVert \pdomega \, \nsigma \rVert}.
\end{align}
To complete the local basis, we define
\[ \tgamma = \ngamma \times \ndomega. \]
\end{definition}
Obviously, $\{ \ngamma, \ndomega, \tgamma \}$ form a right-handed orthonormal basis of $\RR^3$. The vector $\tgamma$ is tangential to the interface $\Sigma$ and tangential to the contact line curve $\Gamma$.\newline
\newline
Moreover, it is useful to define an interface tangent vector $\tau$ in the plane spanned by $\ngamma$ and $\ndomega$. The expansions of $\nsigma$ and $\tau$ are given by
 \begin{equation}
 \label{eqn:expansion_tau_sigma}
 \begin{aligned}
 \tau &= - \cos \theta \, \ngamma -\sin \theta \, \ndomega,\\
 \nsigma &= \sin \theta \, \ngamma - \cos \theta \, \ndomega.
 \end{aligned}
 \end{equation}
Note that $\tau$ is normalized, orthogonal to $\nsigma$ and it is pointing into the domain $\Omega$, since
 \[ \inproduct{\tau}{\ndomega} = - \sin \theta \leq 0. \]

\begin{definition}[Normal and contact line velocity]
To formulate the kinematic boundary condition, we need the notion of normal and contact line velocities.
\begin{enumerate}[(i)]
 \item Let $x^\Sigma: I \rightarrow \RR^3$ be a $\mathcal{C}^1$-curve on $\gr \overline\Sigma$, i.e.\ $(t,x^\Sigma(t)) \in \gr \overline\Sigma \ \forall t \in I$. Then, for $t_0 \in I$ and $x_0 = x^\Sigma(t_0) \in \overline{\Sigma}(t_0)$, the \emph{normal velocity} is defined as\footnote{Note that such curves can always be constructed with the help of a local $\mathcal{C}^1$-parametrization. In the case of a hypersurface with boundary, such a parametrization is defined over the upper half ball
 \[ B_\varepsilon^{n,+}(0) := \{x \in \RR^n: \, \norm{x} < \varepsilon, \ x_n \geq 0 \}. \]
 Moreover, it can be shown that the normal velocity is well-defined, i.e.\ its value is independent of the choice of the curve (see, e.g., \cite{Pruss.2016}, chapter 5.2). 
 }
 \begin{align}
 \normalspeed(t_0,x_0) := \inproduct{\dot{x}^\Sigma(t_0)}{\nsigma(t_0,x_0)}.
 \end{align}
 \item Let $x^\Gamma: I \rightarrow \RR^3$ be a $\mathcal{C}^1$-curve on $\gr\Gamma$, i.e.\ $(t,x^\Gamma(t)) \in \gr\Gamma \ \forall t \in I$. Then, for $t_0 \in I$ and $x_0 = x^\Gamma(t_0) \in \Gamma(t_0)$, the \emph{contact line velocity} is defined as
 \begin{align}
 \clspeed(t_0,x_0) := \inproduct{\dot{x}^\Gamma(t_0)}{\ngamma(t_0,x_0)}.
 \end{align}
 If $\clspeed > 0$ ($\clspeed < 0$), the contact line is said to be advancing (receding). 
\end{enumerate}
\end{definition}
Note that $(t,x^\Gamma(t)) \in \gr\Gamma \subset I \times \partial\Omega$ implies $\dot{x}^\Gamma~\cdot~\ndomega~=~0$ and, hence,
\begin{align*}
&\normalspeed(t_0,x_0) = \inproduct{\dot{x}^\Gamma(t_0)}{\nsigma(t_0,x_0)} \\
&= \inproduct{\dot{x}^\Gamma(t_0)}{\sin \theta \, \ngamma(t_0,x_0)  - \cos \theta \, \ndomega(t_0,x_0) }\nonumber \\
&= \sin \theta \inproduct{\dot{x}^\Gamma(t_0)}{\ngamma(t_0,x_0)} = \sin \theta \, \clspeed(t_0,x_0).
\end{align*}
Therefore, we obtain the important relation
\begin{align}
\label{eqn:vsigma_vgamma}
\normalspeed = \sin \theta \, \clspeed \quad \text{on} \ \gr\Gamma. 
\end{align}

To formulate the two-phase flow model, we need the notion of the \emph{jump} of a quantity across the interface $\Sigma$. 
\begin{definition}[Jump across $\Sigma$]
\label{def:jump}
Fixing an interface configuration $\Sigma$, we define the space $\regular(\Omega,\Sigma)$ of continuous functions on $\Omega^\pm$, admitting continuous extensions to $\overline{\Omega^\pm}$, i.e.
\begin{align*}
\regular(\Omega,\Sigma) := \{ \psi \in \mathcal{C}(\Omega\setminus\Sigma), \ \exists \psi^\pm \in \mathcal{C}(\overline{\Omega^\pm}) \\
\text{s.t.} \ \psi_{|\Omega^\pm} = \psi^\pm \}.
\end{align*}
The jump of $\psi$ across the interface $\Sigma$ at $x \in \overline{\Sigma}$ is defined as
\begin{align*}
\jump{\psi}(x) := \lim_{n \rightarrow \infty} (\psi^+(x_n) - \psi^-(y_n)),
\end{align*}
where $(x_n)_{n \in \NN} \subset \overline{\Omega^+}$ and $(y_n)_{n \in \NN} \subset \overline{\Omega^-}$ are sequences with
\[ \lim_{n \rightarrow \infty} x_n = \lim_{n \rightarrow \infty} y_n = x. \]
\end{definition}
By definition of $\regular$, the jump of $\psi$ does not depend on the choice of sequences. Note that, away from the boundary $\partial\Omega$, the jump of a quantity $\psi$ may equivalently be expressed as
\[ \jump{\psi}(x) = \lim_{h \rightarrow 0^+} (\psi(x+h\nsigma)-\psi(x-h\nsigma)).  \]
For $\psi \in \regular(\Omega, \Sigma)$ it follows directly from the definition that $\jump{\psi}$ is a continuous function on $\overline{\Sigma}$.

\subsection{Energy balance}
We recall the basic modeling assumptions leading to the ``standard model'' in the framework of the sharp interface two-phase Navier Stokes equations  (for the modeling see \cite{Ishii.2011}, \cite{Slattery.1999}, \cite{Edwards.1991}, \cite{Pruss.2016}). It is assumed that the flow in the bulk phases is incompressible and no mass is transferred across the fluid-fluid and the fluid-solid interface. As a further simplification, it is assumed that also the tangential component of the velocity is continuous. These assumptions lead to the formulation
\begin{align}
\rho \, \DDT{v} = \divergence{T}, \quad \nabla \cdot v = 0 \quad &\text{in} \ \Omega\setminus\Sigma(t), \label{eqn:full_system_1} \\
\jump{v} = 0, \quad \normalspeed = \inproduct{v}{\nsigma} \quad &\text{on} \ \Sigma(t), \label{eqn:full_system_2}\\
\clspeed = \inproduct{v}{\ngamma} \quad &\text{on} \ \Gamma(t),\label{eqn:full_system_3}\\
\inproduct{v}{\ndomega}=0 \quad &\text{on} \ \partial\Omega\setminus\Gamma(t), \label{eqn:full_system_4}
\end{align}
where $T=T^\transpose$ is the Cauchy stress tensor\footnote{Here it is assumed that the fluid particles do not carry angular momentum.}.\newline
\newline
To close the model, we consider the energy of the system.

\begin{definition}
In the simplest case, the total available energy of the system is defined as (see, e.g., \cite{Ren.2007}, \cite{E.2011})
\begin{align}
\energy(t) = \int_{\Omega\setminus\Sigma(t)} \frac{\rho v^2}{2} \, dV + \int_{\Sigma(t)} \, \sigma \, dA + \int_{\wet(t)} \sigmawet \, dA,
\end{align}
where $\wet(t):= \overline{\Omega^-(t)}\cap\partial\Omega$ is the wetted area at time $t$ and $\sigma, \, \sigmawet := \sigma_1-\sigma_2$ are the specific energies of the fluid-fluid interface and the wetted surface (relative to the ``dry'' surface).\newline
\newline
Assuming \emph{constant} surface energies $\sigma, \, \sigmawet$ with $\sigma > 0$ and $|\sigmawet| < \sigma$ we \emph{define} an angle $\thetaeq \in (0,\pi)$ by the relation
\begin{align}
\label{eqn:young_defines_theta0}
\sigma \cos \thetaeq + \sigmawet = 0.
\end{align}
\end{definition}

A direct calculation shows the following result for the energy balance, see \cite{Schweizer.2001}, \cite{Ren.2007} for a proof.
\begin{theorem}
\label{theorem:energy_dissipation}
Let $\sigma, \sigmawet$ be constant with $\sigma > 0, \, |\sigmawet| < \sigma$ and $(v,p,\gr\overline\Sigma)$ be a sufficiently regular (classical) solution of the system \eqref{eqn:full_system_1} - \eqref{eqn:full_system_4}. Then
\begin{align}
\label{eqn:energy_dissipation_unclosed_model}
\ddt{\energy} = - &2 \int_{\Omega\setminus\Sigma(t)} D : T \, dV + \int_{\partial\Omega} \inproduct{v}{T\ndomega} \, dA\nonumber\\
- &\int_{\Sigma(t)}(\jump{T} \nsigma + \sigma\kappa \nsigma) \cdot v \, dA\nonumber \\  
+ & \, \sigma \int_{\Gamma(t)} (\cos \theta - \cos \thetaeq) \, \clspeed \, dl,
\end{align}
where $\kappa = - \divsigma \nsigma$ denotes the mean curvature of $\Sigma$ and $D=\frac{1}{2}(\nabla v + \nabla v^\transpose)$ is the rate-of-deformation tensor.
\end{theorem}
According to the second law of thermodynamics, closure relations need to be found such that
\[ \ddt{\energy} \leq 0. \]

\subsection{The standard model}
Employing the standard closure for the two-phase Navier-Stokes model for a Newtonian fluid with dynamic viscosity $\visc$ and constant surface tension $\sigma$, i.e.
\begin{align*}
T = -p \mathds{1} + S &= -p \mathds{1} + \visc (\nabla v + (\nabla v)^\transpose),\\
\jump{-T} \nsigma &= \sigma \kappa \nsigma, 
\end{align*}
we obtain
\begin{align}
\label{eqn:energy_dissipation_semiclosed_model}
\ddt{\energy} = &-2 \int_{\Omega\setminus\Sigma(t)} \visc D:D \, dV + \int_{\partial\Omega} \inproduct{v}{S\ndomega} \, dA\nonumber\\
&+ \sigma  \int_{\Gamma(t)} (\cos \theta - \cos \thetaeq) \, \clspeed \, dl. 
\end{align}
Note that the second term vanishes if the usual no-slip condition is kept. However, as pointed out before, this approach does not allow for a moving contact line. Therefore, it is a frequent choice to consider the following generalization.

\begin{remark}[Navier slip condition]
Assuming that no fluid particles can move across the solid-fluid boundary one still requires $v \cdot \ndomega = 0$ on $\partial\Omega$. In this case, the above term can be rewritten as
\[ \int_{\partial\Omega} \inproduct{\pdomega v}{\pdomega S\ndomega} \, dA. \]
Hence a possible choice to make it non-positive is given by
\begin{align}
\label{eqn:navier_slip_beta}
\inproduct{v}{\ndomega} = 0, \ \pdomega S\ndomega =  -\lambda \pdomega v
\end{align}
on $\partial\Omega$ with $\lambda \geq 0$. Note that equation \eqref{eqn:navier_slip_beta} can be understood as a force balance, where $\lambda$ plays the role of a friction coefficient. The no-slip condition is recovered in the limit $\lambda \rightarrow \infty$, while the case $\lambda = 0$ is known as the \emph{free-slip condition}. The quantity 
\begin{align*}
L = \frac{\visc}{\lambda}
\end{align*}
has the dimension of a length and is called \emph{slip length}. Note that in the two-phase case, the parameters $\lambda$ and $\visc$ are in general discontinuous across the interface. If $L$ is strictly positive, it may be more convenient to use the inverse slip length
\[ \il = \frac{1}{L} = \frac{\lambda}{\visc}. \]
With the inverse slip length, the Navier condition can be expressed as
\begin{align}
\label{eqn:navier_condition_a}
\il \pdomega v + 2 \pdomega D \, \ndomega = 0.
\end{align}
Note that the slip length may depend on various physical parameters of the system including the wettability of the solid and the local shear-rate (see \cite{Neto.2005},\cite{Lauga2007} for a discussion of boundary slip). In the present paper, we only assume that the slip length is a positive function admitting one-sided limits at the contact line, i.e.
\[ \il \in \mathcal{C}(\gr\overline{\Omega^+};[0,\infty))\cap \mathcal{C}(\gr\overline{\Omega^-};[0,\infty)). \]
\end{remark}

\begin{remark}[Contact angle boundary condition]
\label{remark:advancing_for_large_ca}
It remains to close the last term in \eqref{eqn:energy_dissipation_semiclosed_model}. A \emph{sufficient} condition to ensure energy dissipation is to require that 
\begin{align}
\label{eqn:advancing_receding_condition}
\clspeed (\theta - \thetaeq) \geq 0.
\end{align}
This may be achieved by setting\footnote{Or by setting $\clspeed = \gtheta(\theta)$ with $\gtheta(\thetaeq)=0, \ \gtheta(\theta)(\theta-\thetaeq) \geq 0$, which is more convenient if contact angle hysteresis is present.}
\begin{align}
\theta = \ftheta(\clspeed) 
\end{align}
with some function $\ftheta$ satisfying 
\begin{align}
\label{eqn:thermodynamic_conditions_f_g}
\ftheta(0)=\thetaeq, \quad \clspeed \, (\ftheta(\clspeed) - \thetaeq) \geq 0. 
\end{align}
So, in the absence of external forces, the contact line should only advance if the contact angle is above or equal to the equilibrium value defined by the Young equation \eqref{eqn:young} (and vice versa). This is reasonable if we think of the example of a spreading droplet with an initial contact angle larger than the equilibrium value (see Figure~\ref{fig:spreading_droplet}). We expect the contact line to advance \emph{in order to} lower the contact angle and to drive the system towards equilibrium.
\end{remark}
\begin{figure}[hbt]
 \centering
 \includegraphics[width=6cm]{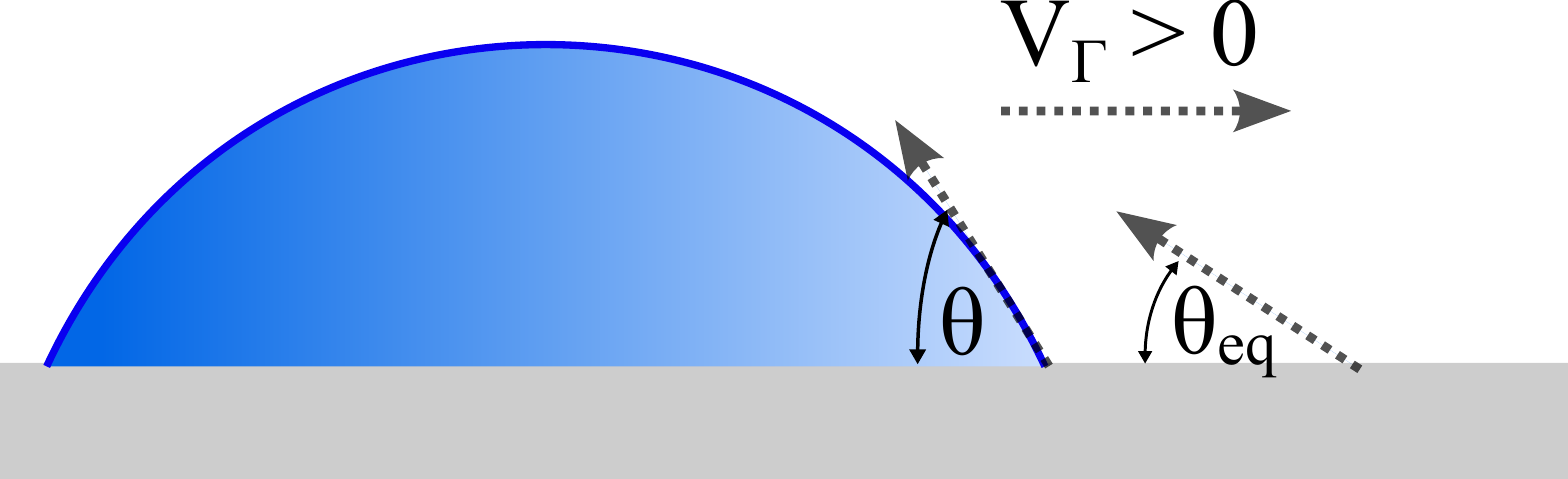}
 \caption{Spreading droplet with an advancing contact line.}
 \label{fig:spreading_droplet}
\end{figure}

\paragraph{Standard model for moving contact lines:}
To summarize, we obtained the ``\emph{standard model for moving contact lines}'' for incompressible two-phase flows with surface tension in the simplest possible case. This is a purely hydrodynamic model without any transfer processes of heat or mass.
\begin{align*}
\rho \DDT{v} - \visc \Delta v + \nabla p = 0, \ \nabla \cdot v = 0 \ \quad &\text{in} \ \Omega\setminus\Sigma(t),\\
\jump{v} = 0, \quad \jump{p \mathds{1} - S} \, \nsigma = \sigma \kappa \nsigma \quad &\text{on} \ \Sigma(t),\\
\inproduct{v}{\ndomega}=0, \ \il \, \pdomega v + 2 \pdomega D \ndomega = 0 \quad &\text{on} \ \partial\Omega\setminus\Gamma(t),\\
\normalspeed = \inproduct{v}{\nsigma} \quad &\text{on} \ \Sigma(t), \\
\clspeed = \inproduct{v}{\ngamma}, \quad \theta = \ftheta(\clspeed) \quad &\text{on} \ \Gamma(t).
\end{align*}
To ensure energy dissipation, we further require
\begin{align*}
\visc \geq 0, \ \il \geq 0, \ \sigma \geq 0, \ \clspeed \,(\ftheta(\clspeed) - \thetaeq) \geq 0. 
\end{align*}

\begin{remark}[Motivation]
Consider a bounded domain $\Omega$ (in $\mathds{R}^2$ or $\mathds{R}^3$) with a smooth boundary $\partial\Omega$. Define a passively advected interface $\Sigma(t)$ as the zero contour of some \emph{level set} function $\phi$, i.e. 
\[ \Sigma(t) = \{ x \in \Omega: \, \phi(t,x) = 0 \}, \]
where $\phi$ satisfies the transport equation
\begin{equation}
\label{eqn:scalar_transport}
\begin{aligned}
\partial_t \phi + v \cdot \nabla \phi &= 0, \quad t > 0, \, x \in \Omega, \\
\phi(0,x) &= \phi_0(x), \quad x \in \Omega
\end{aligned}
\end{equation}
with a given velocity field $v$. The initial value problem \eqref{eqn:scalar_transport} is well-posed if $v$ is sufficiently regular and tangential to the boundary $\partial\Omega$, i.e.
\begin{align*}
v \cdot \ndomega = 0 \quad \text{on} \quad \partial\Omega.
\end{align*}
This can be shown by the method of characteristics (see \cite{evans10}). In particular, there is no boundary condition for $\phi$, i.e.\ no contact angle can be prescribed. However, many methods for the simulation of flows with contact lines use a \emph{time-explicit} discretization of the interface transport equation, i.e.
\begin{align}
\frac{\phi^{n+1}-\phi^n}{\Delta t} + v^n \cdot \nabla \phi^n =  0,
\end{align}
where $v^n$ is tangential to the boundary, and impose a ``contact angle boundary condition'', i.e.\ a boundary condition for $\phi^{n+1}$. But from a mathematical point of view, there is no degree of freedom left allowing to impose such a condition. In fact, the evolution of the contact angle is fully determined by the velocity $v^n$ at time step $n$ and can only be ``adjusted'' or ``corrected'' afterwards. This observation clearly underlines the need for further understanding.
\end{remark}

\section[A kinematic evolution equation]{A kinematic evolution\\ equation}
\label{section:geometrical_evolution_equation}
Let us now derive the contact angle evolution equation. It is a purely kinematic result which, under certain regularity assumptions, follows directly from the kinematic boundary conditions on $\gr\Sigma$ and $\gr\Gamma$.  Note that, in the general case, the interface velocity is only defined on the interface itself, hence trajectories need to fulfill the time-dependent constraint $x(t) \in \overline{\Sigma}(t)$ on $I$, or, equivalently, $(t,x(t)) \in \gr \overline\Sigma$ on $I$. This allows to apply the result to models where the interface is not a material interface, but moves with its \emph{own} velocity $\vsigma$ different from the fluid velocity $v$. This is the case if phase change phenomena or interfacial mass densities are present.

\subsection[Preliminaries: Solutions in closed sets]{Preliminaries: Solutions in\\ closed sets}
There is a well-established theory for solutions of ODEs in closed subsets of $\RR^d$ (see, e.g., \cite{Deimling.1977},\cite{Deimling.1992},\cite{Aubin.2009},\cite{Bothe.1996},\cite{Bothe.2003},\cite{Nagumo.1942},\cite{Brezis.1970},\cite{Amann.1990}). We use this theory for (local) existence and uniqueness of trajectories on the moving hypersurface with boundary. A key definition of the theory is the following:\newline
\newline
Let $K$ a closed subset of $\RR^d$. Then for $y \in K$ the \emph{Bouligand contingent cone} is the set\footnote{The \emph{distance} of a point $x \in \RR^d$ to a subset $A \subseteq \RR^d$ is defined as $\dist(x,A) := \inf\{ \norm{x-a}: \, a \in A \}$.}
\begin{align}
T_K(y) := \{ z \in \RR^d: \liminf_{h \rightarrow 0^+}\, \frac{1}{h} \dist(y+hz, K) = 0 \}.
\end{align}
An element of $T_K(y)$ is said to be \emph{subtangential} to $K$ at $y$. If $K$ is a $\mathcal{C}^1$-submanifold the contingent cone coincides with the tangent space of the submanifold. For a boundary point of a submanifold with boundary, it provides a proper generalization. Note that if $\tilde{K}\subseteq K$ is a closed subset of $K$ then\footnote{This follows from the fact that\newline $\dist(y+hz,K) \leq \dist(y+hz, \tilde{K})$.}
\begin{align}
\label{eqn:cone_relation}
T_{\tilde{K}}(y) \subseteq T_K(y) \quad \forall \ y \in \tilde{K}.
\end{align}
The following result is a special case of Theorem~4.2 in \cite{Deimling.1977} and states that a subtangential and Lipschitz continuous map induces a local semiflow on $K$. If, in addition, we have that both $f(y)$ and $-f(y)$ are subtangential, we obtain a local \emph{flow} on $K$ (in forward and backward direction).\newline
\newline
In the following, $B_r^d(x) := \{ y \in \RR^d: \norm{x-y} < r \}$ denotes the open ball in $\RR^d$ with radius $r$.

\begin{theorem}[see \cite{Deimling.1977}]
\label{thm:solutions_in_closed_sets}
Let $X=\RR^d$, $K \subset X$ be closed, $y_0 \in K$, $K_r := K \cap \overline{B^d_r(y_0)}$ and $f: K_r \rightarrow X$ be Lipschitz continuous with $|f(y)|\leq c$ and
\begin{align}
f(y) \in T_K(y) \quad \forall \ y \in K_r. 
\end{align}
Then the initial value problem
\[y'(s) = f(y(s)), \quad y(0) = y_0 \]
has a unique solution on $[0,r/c]$ with values in $K_r$.
\end{theorem}

\subsection[Trajectories on the moving hypersurface]{Trajectories on the moving\\ hypersurface}
To formulate the kinematic evolution equation, we need the notion of a surface Lagrangian derivative on the moving hypersurface with boundary. We, therefore, consider the flow on $\gr\overline\Sigma$ generated by a consistent interfacial velocity field $\vsigma$.

\begin{lemma}[Trajectories on $\gr\overline\Sigma$]
\label{lemma:trajectories_on_moving_hypersurface}
Let $\gr\overline\Sigma$ be a $\mathcal{C}^{1,2}$-family of moving hypersurfaces with boundary and $\vsigma \in \mathcal{C}^1(\gr\overline\Sigma)$ be a velocity field with
\begin{equation}
\label{eqn:kinematic_conditions_contact_line}
\begin{aligned}
\normalspeed &= \inproduct{\vsigma}{\nsigma} \ \text{on} \ \gr\Sigma,\\
\clspeed &= \inproduct{\vsigma}{\ngamma} \ \text{on} \ \gr\Gamma.
\end{aligned}
\end{equation}
Then the initial value problem
\begin{equation}
\label{eqn:flow_ode}
\begin{aligned}
\ddt{}\, \flowmap{}(t;t_0,x_0) &= (1,\vsigma(\flowmap{}(t;t_0,x_0))),\\
\flowmap{}(t_0;t_0,x_0)&=(t_0,x_0).
\end{aligned}
\end{equation}
is locally uniquely solvable on $\gr\overline\Sigma$. The solution of \eqref{eqn:flow_ode} depends continuously on the initial data $(t_0,x_0)$ and the manifolds $\gr\Sigma$ and $\gr\Gamma$ are invariant subsets for the flow $\Phi$.
\end{lemma}
We call a solution $\flowmap{}(t;t_0,x_0)$ a \emph{trajectory} on the moving hypersurface. Note that due to the structure of $\gr\overline\Sigma$, any solution can be written in the form
\[ \flowmap{}(t;t_0,x_0)=(t,\flowmap{x}(t;t_0,x_0))\]
with $\flowmap{x}(t;t_0,x_0) \in \overline{\Sigma(t)}$. 

\begin{definition}
The \emph{Lagrangian time-derivative} of a quantity $\psi \in \mathcal{C}^1(\gr\overline\Sigma)$ is defined as
\begin{align}
\DSigmaDT{\psi}(t_0,x_0) := \ddt{} \, \psi(\flowmap{}(t;t_0,x_0))\Big|_{t=t_0}.
\end{align}
For an inner point\footnote{Here and in the following, we mean by ``inner point'' an inner point of the \emph{manifold with boundary} $\gr\overline\Sigma$. Clearly, the set $\gr\overline\Sigma$ has no inner points as a subset of $\RR^4$ in the standard topology.} $(t_0,x_0) \in \gr \Sigma$, one may consider the choice
\[ \vsigma := \normalspeed \, \nsigma \]
and write $\thomas$ for the corresponding Lagrangian derivative (also called \emph{Thomas derivative}).
\end{definition}

For the proof of Lemma~\ref{lemma:trajectories_on_moving_hypersurface}, it is useful to give the following characterization of the tangent spaces of the submanifolds $\gr\Sigma$ and $\gr\Gamma$. The proof is given in the Appendix, Lemma~\ref{lemma:tangent_spaces_appendix}.
\begin{lemma}[Tangent spaces]
\label{lemma:tangent_spaces}
The tangent space of $\gr\Sigma$ at the point $(t,x)$ is given by
\begin{align*}
T_{\gr\Sigma}(t,x) = \{ \lambda \, (1,\normalspeed \nsigma(t,x)) + (0,\tau):\\
\lambda \in \RR, \, \tau \in T_{\Sigma(t)}(x) \}.
\end{align*}
Likewise, the tangent space of $\gr\Gamma$ at the point $(t,x)$ is given by
\begin{align*}
T_{\gr\Gamma}(t,x) = \{ \lambda \, (1,\clspeed \ngamma(t,x)) + (0,\tau):\\
\lambda \in \RR, \, \tau \in T_{\Gamma(t)}(t,x) \}.
\end{align*}
\end{lemma}

\begin{proof}[Proof of Lemma~\ref{lemma:trajectories_on_moving_hypersurface}] 
It follows from \eqref{eqn:kinematic_conditions_contact_line} together with \eqref{eqn:vsigma_vgamma} that
\begin{align} 
\label{eqn:vsigma_1}
\vsigma(t,x) &= \normalspeed \nsigma(t,x) + \mathcal{P}_{\Sigma(t)} \vsigma(t,x)
\end{align}
if $(t,x) \in \gr \Sigma$ and, similarly,
\begin{align}
\label{eqn:vsigma_2}
\vsigma(t,x) &= \clspeed\ngamma(t,x) + (\inproduct{\vsigma}{\tgamma}\tgamma)(t,x)
\end{align}
if $(t,x) \in \gr\Gamma$. Hence, the field $f\in \mathcal{C}^1(\gr{\overline\Sigma};\RR^4)$ defined by
\[ f(t,x) := (1,\vsigma(t,x)) \]
is an element of the tangent space $T_{\gr\Sigma}(t,x)$ or $T_{\gr\Gamma}(t,x)$, respectively (cf. Lemma~\ref{lemma:tangent_spaces}). In order to apply Theorem~\ref{thm:solutions_in_closed_sets}, we set $X:=\RR \times \RR^3$, fix a point $(t_0,x_0) \in \gr\overline\Sigma$ and consider the closed subsets
\begin{align*} 
K^\delta (\overline\Sigma) &:= \bigcup_{t \in [t_0-\delta,t_0+\delta]} \{t\} \times \overline\Sigma(t) \subset \gr\overline\Sigma,\\
K^\delta_r (\overline\Sigma) &:= K^\delta (\overline\Sigma) \cap \overline{B^4_r}(t_0,x_0),  \\
K^\delta (\Gamma) &:= \bigcup_{t \in [t_0-\delta,t_0+\delta]} \{t\} \times \Gamma(t) \subset K^\delta (\overline\Sigma),\\
K^\delta_r (\Gamma) &:= K^\delta (\Gamma) \cap \overline{B^4_r}(t_0,x_0)
\end{align*}
for $\delta > 0$ sufficiently small. 

\begin{lemma}
Under the above assumptions there is $r>0$ such that
\begin{align*} 
\pm(1,\vsigma(t,x)) \in T_{K^\delta(\overline\Sigma)}(t,x) \quad \forall (t,x) \in K^\delta_r(\overline\Sigma) \cap \gr\Sigma, \\
\pm(1,\vsigma(t,x)) \in T_{K^\delta(\Gamma)}(t,x)  \quad \forall (t,x) \in K^\delta_r(\overline\Sigma) \cap \gr{\Gamma}.
\end{align*}
Since $T_{K^\delta(\Gamma)} \subset T_{K^\delta(\overline\Sigma)}(t,x)$, it also holds that
\begin{align}
\label{eqn:cone_condition_3}
\pm(1,\vsigma(t,x)) \in T_{K^\delta(\overline\Sigma)}(t,x)
\end{align}
for all $(t,x) \in K^\delta_r(\overline\Sigma)$.
\end{lemma}
\begin{proof}
We choose $r>0$ such that $K^\delta_r(\overline\Sigma) \subseteq K^{\delta/2}(\overline\Sigma) \subset K^{\delta}(\overline\Sigma)$. Therefore, we do not have to consider the boundary cases $t = t_0 \pm \delta$.\newline
\newline
Let $(t,x) \in K^\delta_r(\overline\Sigma) \cap \gr\Sigma$. In this case, the vector $(1,\vsigma(t,x))$ is an element of the tangent space of the manifold $\gr\Sigma$. This follows from \eqref{eqn:vsigma_1} and Lemma~\ref{lemma:tangent_spaces}. By definition this means that there is an open interval $I \ni 0$ and a $\mathcal{C}^1$-curve $\gamma: I \rightarrow \gr\Sigma$ such that
 \[ \gamma(0) = (t,x), \quad \gamma'(0) = (1,\vsigma(t,x)). \]
 Clearly, by restriction to a smaller open interval, one can always achieve $\gamma \in \mathcal{C}^1(\tilde{I};K^\delta(\overline\Sigma))$. Therefore, we have
 \begin{align*}
 &\dist[(t,x) + s(1,\vsigma(t,x)), K^\delta(\overline\Sigma)]\\ 
 \leq \ &|(t,x) + s(1,\vsigma(t,x)) - \gamma(s)| \\
 \leq \ &|(t,x) + s(1,\vsigma(t,x)) - \gamma(0) - \gamma'(0)s + o(|s|)| \\
 = \ &|o(|s|)| \quad \text{as} \ s \rightarrow 0.
 \end{align*}
 Note that this also means that
 \[ \dist[(t,x) - s(1,\vsigma(t,x)), K^\delta(\overline\Sigma)] = |o(|s|)| \]
 as $s \rightarrow 0$. Hence it follows that $\pm(1,\vsigma) \in T_{K^\delta(\overline\Sigma)}$.\newline
\newline
 Let $(t,x) \in K^\delta_r(\Gamma) = K^\delta_r(\overline\Sigma) \cap \gr{\Gamma}$. Since $(1,\vsigma(t,x))$ is an element of the tangent space of the manifold $\gr\Gamma$, there is an open interval $I \ni 0$ and a $\mathcal{C}^1$-curve $\gamma: I \rightarrow K^\delta(\Gamma)$ such that
\[ \gamma(0) = (t,x), \quad \gamma'(0) = (1,\vsigma(t,x)). \]
With the same argument as above, we obtain
\[ \dist[(t,x) \pm s(1,\vsigma(t,x)), K^\delta(\Gamma)] = |o(|s|)| \]
as $s \rightarrow 0$.\qedhere
\end{proof}

Since we have that
\[ \pm(1,\vsigma(t,x)) \in T_{K^\delta(\Gamma)} \quad \forall (t,x) \in K^\delta_r(\Gamma), \]
Theorem~\ref{thm:solutions_in_closed_sets} also implies that the boundary $\gr\Gamma$ is an \emph{invariant subset}, i.e.\ any trajectory starting in the subset $\gr\Gamma$ stays in this subset (in both forward and backward direction). Since $\Phi$ is a \emph{flow} on $\gr\overline\Sigma$, it follows that also the \emph{interior} $\gr\Sigma$ is an invariant subset (see \cite{Amann.1990}). Since we assume $\vsigma \in \mathcal{C}^1(\gr\overline\Sigma)$, standard arguments show that the solution of \eqref{eqn:flow_ode} depends continuously on the initial data $(t_0,x_0)$.
\end{proof}

\subsection{Contact angle evolution equation}
\begin{theorem}[Evolution of the contact angle]
\label{theorem:contact_angle_evolution}
Consider a $\mathcal{C}^{1,2}$-family of moving hypersurfaces with boundary and a consistent velocity field $\vsigma \in \mathcal{C}^1(\gr \overline{\Sigma})$ with
\begin{equation*}
\begin{aligned}
\normalspeed &= \inproduct{\vsigma}{\nsigma} \ \text{on} \ \gr\Sigma,\\
\clspeed &= \inproduct{\vsigma}{\ngamma} \ \text{on} \ \gr\Gamma.
\end{aligned}
\end{equation*}
Let $\Omega$ be a half-space such that $\ndomega$ is constant on the boundary and let $\theta \in (0,\pi)$. Then the time derivative of the contact angle on $\gr \Gamma$ obeys the evolution equation
\begin{align}
\label{eqn:main_result}
\boxed{\DSigmaDT{\theta} = \inproduct{\partial_\tau \vsigma}{\nsigma},}
\end{align}
where $\tau = - \cos \theta \, \ngamma - \sin \theta \, \ndomega$.
\end{theorem}

\begin{remark}
\begin{enumerate}[(i)]
 \item There is a short way to (formally) derive the kinematic evolution equation \eqref{eqn:normal_vector_evolution} using the level set formulation. For details see \cite{Fricke.2018}.
 \item Since $\vsigma \cdot \ndomega = 0$ on $\partial\Omega$, equation \eqref{eqn:main_result} may be reformulated as
 \begin{align*} 
 \DSigmaDT{\theta} = \partial_\tau \normalspeed + \cos \theta \, \clspeed \inproduct{\tau}{\partial_\tau \nsigma} \\
 - \inproduct{\vsigma}{\tgamma}\inproduct{\tgamma}{\partial_\tau \nsigma}. 
 \end{align*}
 In particular, for the two-dimensional case we obtain
 \[ \DSigmaDT{\theta} = \partial_\tau \normalspeed - \kappa \cos \theta \, \clspeed. \]
 In a frame of reference, where the contact line is at rest (i.e. $\clspeed = 0$), the latter formula reduces to
 \[ \DSigmaDT{\theta} = \partial_\tau \normalspeed. \]
 \item Note that for $\theta \rightarrow 0$ or $\theta \rightarrow \pi$, the interface tangent vector $\tau$ becomes tangential to $\partial\Omega$ and $\nsigma \rightarrow \pm \ndomega$. Therefore, we obtain in the limit
 \[ \DSigmaDT{\theta}\Big|_{\theta=0} = \DSigmaDT{\theta}\Big|_{\theta=\pi} = 0.\]
 In the following, we restrict ourselves to the case of \emph{partial wetting}, i.e. 
 \[ 0 < \theta < \pi.\]
\end{enumerate}
\end{remark}

To prove Theorem~\ref{theorem:contact_angle_evolution}, it is useful to first consider the evolution of the normal vector.
\begin{theorem}[Evolution of the normal vector]
\label{theorem:normal_vector_evolution}
Consider a $\mathcal{C}^{1,2}$-family of moving hypersurfaces and a consistent velocity field $\vsigma \in \mathcal{C}^1(\gr \Sigma)$, i.e. such that
\begin{align}
\label{eqn:kinematic_condition}
\normalspeed = \inproduct{\vsigma}{\nsigma} \ \text{on} \ \gr \Sigma.
\end{align}
Then the evolution of the interface normal vector on $\gr \Sigma$ obeys the evolution equation
\begin{align}
\label{eqn:normal_vector_evolution}
\DSigmaDT{\nsigma} = - \sum_{k=1}^2 \inproduct{\partial_{\tau_k} \vsigma}{\nsigma} \tau_k,
\end{align}
where $\{\tau_1,\tau_2\}$ is an orthonormal basis of $T_{\Sigma(t_0)}(x_0)$.
\end{theorem}
\begin{remark}
Since $\normalspeed(t,\cdot) = \inproduct{\vsigma(t,\cdot)}{\nsigma(t,\cdot)} \in \mathcal{C}^1(\Sigma(t))$, equation \eqref{eqn:normal_vector_evolution} can be written as
\begin{align*}
\DSigmaDT{} \nsigma &= \sum_{k=1}^2 \left( - \tau_k \partial_{\tau_k} \normalspeed + \inproduct{(\vsigma)_\parallel}{\partial_{\tau_k} \nsigma} \, \tau_k \right) \nonumber \\
&= - \nablasigma \normalspeed + \sum_{k=1}^2 \inproduct{(\vsigma)_\parallel}{\partial_{\tau_k} \nsigma} \, \tau_k.
\end{align*}
In particular, for $\vsigma(t,x) := \normalspeed(t,x) \, \nsigma(t,x)$ we obtain (in agreement with \cite{Kimura.2008}, Theorem 5.15)
\begin{align}
\thomas \nsigma = - \nablasigma \normalspeed.
\end{align}
With this notation, we may express \eqref{eqn:normal_vector_evolution} as
\begin{align*}
\DSigmaDT{} \nsigma &= \thomas \nsigma + (\nablasigma \nsigma) (\vsigma)_\parallel\\
&= - \nablasigma \normalspeed + \norm{(\vsigma)_\parallel} \, \partial_{w} \nsigma,
\end{align*}
where $w:= (\vsigma)_\parallel/\norm{(\vsigma)_\parallel}$ (for $(\vsigma)_\parallel \neq 0$).
\end{remark}

\paragraph{Preliminaries for the proof:}
In order to prove  Theorem~\ref{theorem:normal_vector_evolution}, we need a \emph{continuously differentiable} dependence of the trajectories $\flowmap{}(\cdot,t_0,x_0)$ on the initial position $x_0 \in \Sigma(t_0)$. To this end, we construct a $\mathcal{C}^1$-extension of the velocity field $\vsigma$ to an open neighborhood of $(t_0,x_0)$ in $\RR^4$, which still leaves $\gr\Sigma$ invariant. This construction allows to obtain the $\mathcal{C}^1$-dependence on the initial position from standard ODE theory. To show the following Lemma, it is helpful to use a special type of local parametrization for $\gr\Sigma$ which is constructed in the Appendix.

\newcommand{\tube}{X}
\begin{lemma}[Signed distance function]
\label{lemma:signed_distance_old}
Let $\{\Sigma(t)\}_{t \in I}$ be a $\mathcal{C}^{1,2}$-family of moving hypersurfaces and $(t_0,x_0)$ be an inner point of $\move=\gr\Sigma$. Then there exists an open neighborhood $U \subset \RR^4$ of $(t_0,x_0)$ and $\varepsilon > 0$ such that the map
\begin{align*}
\tube: (\move \cap U) \times (-\varepsilon, \varepsilon) \rightarrow \RR^4,\\ 
\tube(t,x,h) := (t, x + h \nsigma(t,x))\nonumber
\end{align*}
is a diffeomorphism onto its image
\begin{align*}
\mathcal{N}^\varepsilon := \tube((\move \cap U) \times (-\varepsilon \times \varepsilon)) \subset \RR^4 ,
\end{align*}
i.e.\ $\tube$ is invertible there and both $\tube$ and $\tube^{-1}$ are $\mathcal{C}^1$. The inverse function has the form
\begin{align*}
\tube^{-1}(t,x) = (\pi(t,x),d(t,x))
\end{align*}
with $\mathcal{C}^1$-functions $\pi$ and $d$ on $\mathcal{N}^\varepsilon$.
\end{lemma}
The set $\mathcal{N}^\varepsilon$ is called ``tubular neighborhood'' for $\move$ at the point $(t_0,x_0)$. The function $d$ is the \emph{signed distance} to $\move$ and $\pi$ is the associated \emph{projection operator}. For a fixed hypersurface $\Sigma$, this result is well-known (see, e.g., \cite{Gilbarg.2001},\cite{Pruss.2016}). The above time-dependent result is already stated without details of the proof in  \cite{Kimura.2008}, Lemma 5.12. For completeness, we include a short proof in the Appendix.\\

We now employ Lemma~\ref{lemma:signed_distance_old} and set
\[ v(t,x) := \vsigma(\pi(t,x)) \]
in the tubular neighborhood $\mathcal{N}^\varepsilon$ to construct a local $\mathcal{C}^1$-continuation of $\vsigma$. Note that $v$ generates a local flow map $\tilde{\Phi}$ in an open neighborhood of $(t_0,x_0) \in \RR^4$ by means of \eqref{eqn:flow_ode}. The moving hypersurface $\gr\Sigma$ is \emph{invariant} with respect to $\tilde{\Phi}$ because of the consistency conditions \eqref{eqn:kinematic_conditions_contact_line}. Hence we drop the tilde notation in the following. It is well-known from classical ODE theory that a $\mathcal{C}^1$-right hand side yields a \emph{continuously differentiable} dependence on the initial data. Therefore, we have the following result.

\begin{lemma}[Regularity of the flow map]
\label{lemma:regularity_of_the_flow_map}
Let $x_0 \in \Sigma(t_0)$ and $v \in \mathcal{C}^1(U)$ for an open neighborhood $U$ of $(t_0,x_0) \in \RR^4$. Then $\flowmap{}(\cdot;t_0,\cdot)$ is $\mathcal{C}^1$ on an open neighborhood of $(t_0,x_0) \in \RR^4$.
\end{lemma}

\begin{lemma}[Tangent transport]
Under the assumptions of Theorem~\ref{theorem:normal_vector_evolution}, consider an inner point $(t_0,x_0) \in \gr \Sigma$ and a normalized tangent vector $\tau \in \tangentspace{\Sigma(t_0)}(x_0)$. Choose a curve $\gamma^0((-\delta,\delta);\Sigma(t_0))$ such that
\[ \gamma^0(0) = x_0, \quad (\gamma^0)'(0) = \tau. \]
For simplicity let $\norm{(\gamma^0)'}=1$ on $(-\delta,\delta)$. Then the curve is transported by the flow-map according to\footnote{Recall that $\flowmap{}$ has the structure \[ \flowmap{}(t;t_0,x_0)=(t,\flowmap{x}(t;t_0,x_0)).\]}
\begin{align}
\label{eqn:curve_transport}
\gamma(s,t):= \flowmap{x}(t;t_0,\gamma^0(s)). 
\end{align}
Likewise, a time evolution for the (not necessarily normalized) tangent vector is defined by
\begin{align}
\label{eqn:tangent_transport}
\tau(t):= \frac{\partial}{\partial s} \gamma(s,t)\Big|_{s=0}. 
\end{align}
The vector $\tau(t)$ is tangent to $\Sigma(t)$ at the point $\flowmap{x}(t;t_0,x_0)$ since $\gamma^0(\cdot,t) \subset \Sigma(t)$. Moreover, its time derivative is given as
\begin{align}
\label{eqn:tangent_derivative}
\tau'(t_0) = \frac{\partial\vsigma}{\partial \tau(t_0)}(t_0,x_0).
\end{align}
\end{lemma}
\begin{proof}
By definition, we have
\[ \tau'(t_0) = \pddt{} \left( \frac{\partial}{\partial s} \flowmap{x}(t;t_0,\gamma^0(s))\Big|_{s=0} \right)\Big|_{t=t_0}. \]
Since $\gamma \in \mathcal{C}^1$ and the second partial derivative
\begin{align*} 
&\frac{\partial}{\partial s} \pddt{} \gamma(s,t) = \frac{\partial}{\partial s} \pddt{}  \flowmap{x}(t;t_0,\gamma^0(s))\\
&= \frac{\partial}{\partial s} \vsigma(t,\flowmap{x}(t;t_0,\gamma^0(s))) \\
&= \nablasigma \vsigma(t,\flowmap{x}(t;t_0,\gamma^0(s))) \cdot \frac{\partial}{\partial s} \flowmap{x}(t;t_0,\gamma^0(s))
\end{align*}
is continuous at $(0,t_0)$, it follows from the Theorem of Schwarz that we can interchange the order of differentiation to obtain
\begin{align*} 
\tau'(t_0) &= \frac{\partial}{\partial s} \left( \frac{\partial}{\partial t} \, \flowmap{x}(t;t_0,\gamma^0(s)) \Big|_{t=t_0} \right)\Big|_{s=0} \\
&= \frac{\partial}{\partial s} \vsigma(t_0,\gamma^0(s))\Big|_{s=0} = \frac{\partial \vsigma}{\partial \tau(t_0)}(t_0,x_0).\qedhere
\end{align*}
\end{proof}

\begin{proof}[\textbf{Proof of Theorem~\ref{theorem:normal_vector_evolution}}]
We choose \emph{two} curves 
\[ \gamma^0_1, \gamma^0_2 \in \mathcal{C}^1((-\delta,\delta);\Sigma(t_0)) \]
such that
\begin{align*}
\gamma^0_1(0) = \gamma^0_2(0) = x_0, \quad (\gamma^0_1)'(0) = \tau_1, \quad (\gamma^0_2)'(0) = \tau_2
\end{align*}
with $|\tau_1| = |\tau_2| = 1$ and $ \nsigma(t_0,x_0) = \tau_1 \times \tau_2.$
The flow map $\flowmap{}$ defines a time-evolution of $\gamma_i$ and of the tangent vectors $\tau_i$ according to \eqref{eqn:curve_transport} and \eqref{eqn:tangent_transport}. As long as $\tau_1$ and $\tau_2$ are linearly independent (i.e., if $\tau_1 \times \tau_2 \neq 0$), it follows that
\[ \nsigma(\flowmap{}(t;t_0,x_0)) = \frac{\tau_1(t) \times \tau_2(t)}{|\tau_1(t) \times \tau_2(t)|} \]
and, in particular,
\begin{align}
\label{eqn:normal_vector_evolution_prelim}
\DSigmaDT{\nsigma}\Big|_{t=t_0} = \ddt{} \, \frac{\tau_1(t) \times \tau_2(t)}{|\tau_1(t) \times \tau_2(t)|}\Big|_{t=t_0}.
\end{align}
Note that the linear independence of $\tau_1(t)$ and $\tau_2(t)$ for $t$ sufficiently close to $t_0$ follows from the initial condition, i.e.
\[ |\tau_1(t_0) \times \tau_2(t_0)| = 1, \]
since $\tau_1(t)$ and $\tau_2(t)$ are continuous. From \eqref{eqn:normal_vector_evolution_prelim} it follows that
\begin{align*} 
\DSigmaDT{} \, \nsigma = &\frac{\tau'_1(t_0) \times \tau_2(t_0) + \tau_1(t_0) \times \tau'_2(t_0)}{|\tau_1(t_0) \times \tau_2(t_0)|}\\
&- \frac{\tau_1(t_0) \times \tau_2(t_0)}{|\tau_1(t_0) \times \tau_2(t_0)|^2} \, \ddt{}|\tau_1(t) \times \tau_2(t)|_{t=t_0}.
\end{align*}
From $\tau_1(t_0) \times \tau_2(t_0) = \nsigma$ and $|\nsigma|=1$ we infer
\begin{align*}
\DSigmaDT{} \, \nsigma &= \tau'_1(t_0) \times \tau_2(t_0) + \tau_1(t_0) \times \tau'_2(t_0)\\
&- \nsigma \, \ddt{} |\tau_1(t) \times \tau_2(t)|\Big|_{t=t_0} \\
&= \tau'_1(t_0) \times \tau_2(t_0) + \tau_1(t_0) \times \tau'_2(t_0)\\
&- \nsigma \, \inproduct{\frac{\tau_1(t) \times \tau_2(t)}{|\tau_1(t) \times \tau_2(t)|}}{\ddt{} (\tau_1(t) \times \tau_2(t))}_{t=t_0} \\
&= \psigma (\tau'_1(t_0) \times \tau_2(t_0) + \tau_1(t_0) \times \tau'_2(t_0)),
\end{align*}
where $\psigma:= \mathds{1} - \inproduct{\nsigma}{\cdot}\nsigma$ denotes the orthogonal projection onto $T_{\Sigma}$. Using \eqref{eqn:tangent_derivative}, we conclude
\begin{align*}
\DSigmaDT{} \, \nsigma = \psigma [(\partial_{\tau_1} \vsigma) \times \tau_2 + \tau_1 \times (\partial_{\tau_2} \vsigma)].
\end{align*}
The claim follows by expanding $\partial_{\tau_1} \vsigma$ and $\partial_{\tau_2} \vsigma$ in the basis $\{\tau_1,\tau_2,\nsigma\}$.
\end{proof}

\begin{proof}[\textbf{Proof of Theorem~\ref{theorem:contact_angle_evolution}}]
We first show that equation \eqref{eqn:normal_vector_evolution} also holds at the contact line. As a result, we obtain the evolution of the contact angle.\newline
\newline
For $(t_0,x_0) \in \gr\Gamma$ we choose a sequence of points $(x^k_0)_k \subset \Sigma(t_0)$ such  that $x^k_0$ converges to $x_0$ and consider the trajectories $x^k(t)$ defined by
\begin{align} 
\ddt{} x^k(t) = \vsigma(t,x^k(t)), \quad x^k(t_0) = x^k_0. 
\end{align}
Moreover, we define the limiting trajectory $x(\cdot)$ starting from $x_0$ and running on $\gr\Gamma$. Since $\gr \Sigma$ is invariant under the flow, the evolution equation \eqref{eqn:normal_vector_evolution} holds along $x^k$ for every $k$.  Since $\nsigma \in \mathcal{C}^1(\gr \overline\Sigma)$, one can choose fields $\tau_1, \, \tau_2 \in \mathcal{C}^1(\gr \overline\Sigma)$ such that $\{\tau_1(t,x), \tau_2(t,x)\}$ is an orthonormal basis to the tangent space of $\Sigma(t)$ at the point $x$ such that
\[ \nsigma(t,x) = \tau_1(t,x) \times \tau_2(t,x) \quad \text{on} \ \gr \overline\Sigma. \]
Hence we obtain by integration
\begin{align*} 
\nsigma(t,x^k(t)) &= \nsigma(t_0,x^k_0) \\
&-\sum_{j=1}^2 \int_{t_0}^t [\inproduct{(\nablasigma \vsigma) \, \tau_j}{\nsigma} \tau_j](s,x^k(s))\, ds.
\end{align*}
It follows from the continuous dependence on the initial data that the trajectories converge pointwise to the limiting trajectory, i.e.
\[ \lim_{k \rightarrow \infty} x^k(t) = x(t) \quad \text{on} \ \Gamma(t). \]
Now we pass to the limit and obtain
\begin{align*} 
\nsigma(t,x(t)) &=  \nsigma(t_0,x_0)\\
&- \sum_{j=1}^2 \int_{t_0}^t [\inproduct{(\nablasigma \vsigma) \, \tau_j}{\nsigma} \tau_j](s,x(s))\, ds.  
\end{align*}
Differentiation with respect to $t$ proves that \eqref{eqn:normal_vector_evolution} also holds at the contact line.\newline
\newline
It follows from the definition of $\theta$ that
\[ \DSigmaDT{} \cos \theta = - \DSigmaDT{} \inproduct{\nsigma}{\ndomega}. \]
Since $\ndomega$ is constant, we obtain
\[ - \sin \theta \, \DSigmaDT{\theta} = - \inproduct{\DSigmaDT{\nsigma}}{\ndomega}. \]
We choose 
\[ \tau_1 = \tau = - \cos \theta \, \ngamma - \sin \theta \, \ndomega, \quad \tau_2=\tgamma \] 
and proceed using equation \eqref{eqn:normal_vector_evolution} to arrive at
\begin{align*} 
\sin \theta \DSigmaDT{\theta} = &- \inproduct{\partial_\tau \vsigma}{\nsigma} \inproduct{\tau}{\ndomega} \\
&- \inproduct{\partial_{\tgamma} \vsigma}{\nsigma} \inproduct{\tgamma}{\ndomega}\\
&= \sin \theta \inproduct{\partial_\tau \vsigma}{\nsigma}. 
\end{align*}
This proves the claim since $\theta \in (0,\pi)$.\qedhere
\end{proof}

\section[Contact angle evolution in the framework of the standard model]{Contact angle evolution in the framework of the\\ standard model}
\label{section:ca_evolution_standard_model}
\subsection{Preliminaries}
\begin{definition}[Regularity]
\label{def:regularity}
In the following, we consider an open interval $I$ and the space of functions
\begin{align}
\velspace := \mathcal{C}(I \times \overline\Omega) \cap \mathcal{C}^1(\gr \overline{\Omega^+}) \cap \mathcal{C}^1(\gr \overline{\Omega^-}).
\end{align}
\end{definition}
Note that we assume in particular that the fluid velocities $v^\pm$ are \emph{differentiable} at the contact line and the viscous stress is locally \emph{bounded}. This is a rather strong assumption in contrast to weak solution concepts which allow an \emph{integrable singularity} in the viscous stress as long as the corresponding dissipation rate is finite.

\begin{lemma}[Kinematic conditions]
\label{lemma:kinematic_conditions}
Let $\gr\overline\Sigma$ be a $\mathcal{C}^{1,2}$-family of moving hypersurfaces with boundary, $\theta \in (0,\pi)$ on $\gr\Gamma$ and 
\[ \vsigma \in \mathcal{C}(\gr\overline\Sigma) \]
satisfy
\begin{align}
\normalspeed = \inproduct{\vsigma}{\nsigma} \quad &\text{on} \quad \gr\Sigma,\label{eqn:kinematic_1}\\
0 = \inproduct{\vsigma}{\ndomega} \quad &\text{on} \quad \gr\Gamma.\label{eqn:kinematic_2}
\end{align}
Then the contact line velocity fulfills the kinematic condition
\begin{align}
\label{eqn:kinematic_3}
\clspeed = \inproduct{\vsigma}{\ngamma} \quad \text{on} \quad \gr\Gamma.
\end{align}
\end{lemma}
\begin{proof}
From the relation \eqref{eqn:vsigma_vgamma} it follows that
\begin{align*}
\sin \theta \, \clspeed &= \inproduct{\vsigma}{\nsigma} = \sin \theta \inproduct{\vsigma}{\ngamma} - \cos \theta \inproduct{\vsigma}{\ndomega}\\
&= \sin \theta \inproduct{\vsigma}{\ngamma} \quad \text{on} \quad \gr\Gamma.
\end{align*}
This proves the claim since $\theta \in (0,\pi)$ by assumption.
\end{proof}
Hence \eqref{eqn:kinematic_3} is a consequence of \eqref{eqn:kinematic_1} and \eqref{eqn:kinematic_2} and can be dropped in the problem formulation.

\paragraph{The Continuity Lemma:}
The following Lemma shows an \emph{additional} continuity property for the velocity gradient, which \emph{only} holds at the contact line. Typically, the gradient of the velocity field has a jump, which is controlled by the interfacial transmission conditions.\newline
\newline
Note that we define the gradient of a vector $w$ in Cartesian coordinates as 
\[ (\nabla w)_{i,j} = \frac{\partial w_i}{\partial x_j}.  \]
\begin{lemma}
\label{lemma:continuity_of_grad_v_at_gamma}
Let $\Omega \subset \RR^3$, $0<\theta<\pi$, $v \in \mathcal{C}(\overline{\Omega})$, $\nabla v \, \in \regular(\Omega,\Sigma)$ and
\begin{align*}
\inproduct{v}{\ndomega}=0 \quad \text{on} \ \partial\Omega, \quad \divergence{v} = 0 \quad \text{in} \ \Omega\setminus\Sigma(t),
\end{align*}
where $\partial\Omega$ is the smooth boundary of $\Omega$. Then $\nabla v$ has the following continuity property at the contact line:
\begin{align*}
\jump{\inproduct{\nabla v \, \alpha}{\beta}} = 0 \quad \text{on} \ \Gamma, 
\end{align*}
where $\alpha, \beta$ are arbitrary vectors in the plane spanned by $\ngamma$ and $\ndomega$.
\end{lemma}
\begin{proof}
We consider an arbitrary point on $\Gamma$ and show $\jump{\nabla v \, \tau}=0$ as well as $\jump{\inproduct{\nabla v \, \ngamma}{\ndomega}}=\jump{\inproduct{\nabla v \, \ngamma}{\ngamma}}=0$. This is already sufficient since $\tau$ and $\ngamma$ are linearly independent. Since $v$ is assumed to be continuous across $\Sigma$, the tangential derivatives of $v$ are continuous
 \[ \jump{\nabla v \, \tau} = \jump{\nabla v \, \tgamma} = 0. \]
Since $v$ is tangential to $\partial\Omega$, it follows that
 \[ \inproduct{\nabla v \, \ngamma}{\ndomega} = 0 \quad \Rightarrow \quad \jump{\inproduct{\nabla v \, \ngamma}{\ndomega}} = 0. \]
It remains to show that $\jump{\inproduct{\nabla v \, \ngamma}{\ngamma}} = 0$. Since $v$ is solenoidal, we have
 \[ 0 = \divergence{v} = \inproduct{\nabla v \, \ngamma}{\ngamma} + \inproduct{\nabla v \, \ndomega}{\ndomega} + \inproduct{\nabla v \, \tgamma}{\tgamma}. \]
 Therefore, we can write
 \begin{align*}
 \jump{\inproduct{\nabla v \, \ngamma}{\ngamma}} &= -(\jump{\inproduct{\nabla v \, \ndomega}{\ndomega}}+\jump{\inproduct{\nabla v \, \tgamma}{\tgamma}})\\
 &= - \jump{\inproduct{\nabla v \, \ndomega}{\ndomega}}.
 \end{align*}
 From $\tau = -\cos \theta \, \ngamma - \sin \theta \, \ndomega$ we infer (since $0 < \theta < \pi$)
 \[ \ndomega = -\frac{1}{\sin \theta} \left( \cos \theta \, \ngamma + \tau \right). \]
 This yields
 \begin{align*}
 &\jump{\inproduct{\nabla v \, \ngamma}{\ngamma}} = \frac{1}{\sin \theta} (\cos \theta \jump{\inproduct{\nabla v \, \ngamma}{\ndomega}} \\
 &+ \jump{\inproduct{\nabla v \, \tau}{\ndomega}}) = \frac{\jump{\inproduct{\nabla v \, \tau}{\ndomega}}}{\sin \theta} \\
 &= \frac{\inproduct{\jump{\nabla v \, \tau}}{\ndomega}}{\sin \theta} = 0.\qedhere
 \end{align*}
\end{proof}
Note that in the 2D case the \emph{full} gradient of $v$ is continuous across $\Gamma$. 

\paragraph{On the Navier boundary condition:}
We reconsider the Navier condition \eqref{eqn:navier_condition_a}. By taking the projection onto $\ngamma$ we have
\[ a^\pm \inproduct{v^\pm}{\ngamma} + 2 \inproduct{D^\pm \ndomega}{\ngamma} = 0. \]
If $v$ satisfies the kinematic conditions $v^\pm \cdot \ngamma = \clspeed$, we obtain the jump condition
\begin{align}
\inproduct{\jump{D}\ndomega}{\ngamma}_{|\Gamma} = - \frac{\jump{a} \clspeed}{2}.
\end{align}
Under the assumptions of Lemma~\ref{lemma:continuity_of_grad_v_at_gamma}, we have $\inproduct{\jump{D}\ndomega}{\ngamma}_{|\Gamma}=0$ and hence 
\begin{align}
\label{eqn:jump_condition_slip_length}
\jump{a} \clspeed = 0.
\end{align}
Hence, to allow for a regular solution with $\clspeed \neq 0$, one has to choose $a$ as a continuous function across the contact line, i.e.
  \begin{align*}
  \frac{\lambda^+}{\visc^+}_{|\Gamma} = \frac{\lambda^-}{\visc^-}_{|\Gamma} = \il_{|\Gamma}.  
  \end{align*}
  In this case, we have the relations
  \begin{align}
  \label{eqn:jump_of_lambda}
  \jump{\lambda} = \il \jump{\visc}
  \end{align}
  and 
  \begin{align}
  \label{eqn:navier_gradient_formula}
  2\inproduct{D \ndomega}{\ngamma}_{|\Gamma}=\inproduct{\nabla v \, \ndomega}{\ngamma}_{|\Gamma} = - a \clspeed.
  \end{align}
\subsection{Contact angle evolution}
The following Theorem shows that, for sufficiently regular solutions, $\dot{\theta}$ has a quite simple form  for a large class of models. Note that the equations \eqref{eqn:pde-system-1}-\eqref{eqn:pde-system-4} say nothing about external forces, do not specify the contact angle and the slip length may be a function of space and time. Moreover, we only need the \emph{tangential part} of the transmission condition for the stress. In this sense, the system \eqref{eqn:pde-system-1}-\eqref{eqn:pde-system-4} is not closed but describes a \emph{class} of models.\\

The main idea for the proof is the observation that both the Navier and the interfacial transmission condition are valid at the contact line. A regular classical solution has to satisfy both of them.
\begin{theorem}[]
\label{theorem:ca_evolution_in_standard_model}
Let $\Omega\subset\RR^3$ (or $\Omega\subset\RR^2$) be a half-space with boundary $\partial\Omega$, $\sigma \equiv \text{const}$, $\visc^\pm >0$, $\jump{\visc}\neq 0$, $\il \in \mathcal{C}(\gr\partial\Omega)$ and $(v,\gr\overline\Sigma)$ with $v \in \velspace$, $\gr\overline\Sigma$ a $\mathcal{C}^{1,2}$-family of moving hypersurfaces with boundary, be a classical solution of the PDE-system
\begin{align}
\nabla \cdot v = 0 \quad &\text{in} \ \Omega\setminus\Sigma(t),\label{eqn:pde-system-1} \\
\jump{v} = 0, \quad \mathcal{P}_\Sigma \jump{S} \, \nsigma = 0 \quad &\text{on} \ \Sigma(t), \label{eqn:transmission_condition}\\
\inproduct{v}{\ndomega}=0 \quad &\text{on} \ \partial\Omega\setminus\Gamma(t),\label{eqn:navier_condition_part2}\\
\il \, \pdomega v + 2 \pdomega D \ndomega = 0 \quad &\text{on} \ \partial\Omega\setminus\Gamma(t), \label{eqn:navier_condition}\\
\normalspeed = \inproduct{v}{\nsigma} \quad &\text{on} \ \Sigma(t) \label{eqn:pde-system-4}
\end{align}
with $\theta \in (0,\pi)$ on $\gr\Gamma$. Then the evolution of the contact angle is given by
\begin{align}
\label{eqn:evolution_equation_standard_model}
\frac{D \theta}{D t} = \frac{\il\clspeed}{2} = \frac{\clspeed}{2L}.
\end{align}
Moreover, in the case $\theta = \pi/2$ it holds that
\begin{align}
\label{eqn:pi_over_2_condition}
(\il\clspeed)_{|\theta=\pi/2} = 0, 
\end{align}
which also means that
\[ \DDT{\theta}_{|\theta=\pi/2} = 0. \]
\end{theorem}
\begin{proof} 
Since $v$ is continuous and $v^\pm \in \mathcal{C}^1(\overline{\gr\Omega^\pm})$, we can choose $\vsigma := v^+_{|\gr \overline\Sigma} = v^-_{|\gr \overline\Sigma} \in \mathcal{C}^1(\gr \overline\Sigma)$ and apply Theorem~\ref{theorem:contact_angle_evolution} to obtain
\begin{align}
\label{eqn:theta_evolution_tau_nsigma}
\DDT{\theta} = \inproduct{(\nabla v)^\pm \tau}{\nsigma}. 
\end{align}
Recall that the vectors $\tau$ and $\nsigma$ can be expressed as
\begin{equation}
\label{eqn:tau_nsigma}
\begin{aligned}
\tau &= - \ngamma \cos \theta - \ndomega \sin \theta,\\
\nsigma &= \ngamma \sin \theta  - \ndomega \cos \theta.  
\end{aligned}
\end{equation}
Inserting \eqref{eqn:tau_nsigma} into equation \eqref{eqn:theta_evolution_tau_nsigma} yields
\begin{align*}
\DDT{\theta} = \cos^2 \theta \inproduct{(\nabla v)^\pm \ngamma}{\ndomega}   - \sin^2 \theta \inproduct{(\nabla v)^\pm \ndomega}{\ngamma}\\
+ \sin \theta \cos \theta \left( \inproduct{(\nabla v)^\pm \ndomega}{\ndomega} - \inproduct{(\nabla v)^\pm \ngamma}{\ngamma} \right).
\end{align*}
Notice that the impermeability condition implies that
\[ \inproduct{(\nabla v)^\pm \ngamma}{\ndomega} = 0. \]
Moreover, Lemma~\ref{lemma:continuity_of_grad_v_at_gamma} allows to drop the $\pm$-notation. We now  exploit that both the jump condition and the Navier condition are active at the contact line. Using the relation \eqref{eqn:navier_gradient_formula} it follows from the Navier condition
\begin{align}
\label{eqn:3d_evolution_equation}
\DDT{\theta} = &\sin(\theta)^2 \, \il\clspeed\nonumber\\
&+ \sin \theta \cos \theta \left( \inproduct{\nabla v \, \ndomega}{\ndomega} - \inproduct{\nabla v \, \ngamma}{\ngamma} \right).
\end{align}
Since $\tau$ is tangential to $\Sigma$, it follows from the continuity of the tangential stress component \eqref{eqn:transmission_condition} that
 \begin{align*}
 \inproduct{\jump{S} \nsigma}{\tau} = 0.
 \end{align*}
 Using Lemma~\ref{lemma:continuity_of_grad_v_at_gamma}, we can exploit the continuity property of $\nabla v$ at the contact line to obtain
 \[ 0 = 2 \jump{\visc} \inproduct{D\,\nsigma}{\tau} \quad \Leftrightarrow \quad 0 = \inproduct{D\,\nsigma}{\tau}. \]
 Together with the expansions \eqref{eqn:tau_nsigma} for $\nsigma$ and $\tau$ we obtain
 \begin{align}
 \label{eqn:continuity_of_tangential_stress}
 0 &= \sin \theta \cos \theta \left(- \inproduct{D\ngamma}{\ngamma} + \inproduct{D\ndomega}{\ndomega} \right) \nonumber \\
   &+ (\cos^2 \theta - \sin^2 \theta) \inproduct{D\ndomega}{\ngamma}.
 \end{align}
Using the Navier condition, we can replace the last term to find
  \begin{align}
  \label{eqn:3d_evolution_intermediate_equation}
  0 &= \sin \theta \cos \theta \left(- \inproduct{\nabla v \, \ngamma}{\ngamma} + \inproduct{\nabla v \, \ndomega}{\ndomega} \right) \nonumber\\
    & + (\sin^2 \theta - \cos^2 \theta) \frac{\il \clspeed}{2} \quad \text{on} \ \Gamma.
  \end{align}
  Note that for $\theta = \pi/2$ this reduces to
  \begin{align*}
  \il \clspeed = 0.
  \end{align*}
  The claim follows by inserting equation \eqref{eqn:3d_evolution_intermediate_equation} into the contact angle evolution equation \eqref{eqn:3d_evolution_equation}:
  \begin{align*}
  \DDT{\theta} = &\sin^2(\theta) \, \il\clspeed\\
  &+ \sin \theta \cos \theta (\inproduct{\nabla v \, \ndomega}{\ndomega}-\inproduct{\nabla v \, \ngamma}{\ngamma}) \\
  = &\il \clspeed \left(\sin^2 \theta - \frac{1}{2}(\sin^2 \theta - \cos^2 \theta) \right) = \frac{a \clspeed}{2}.\qedhere 
  \end{align*}
\end{proof}
Note that the incompressibility condition \eqref{eqn:pde-system-1} can be dropped leading to (see Theorem~\ref{theorem:interfacial_slip})
\[ \DDT{\theta} = \frac{\jump{\lambda}}{\jump{\eta}} \frac{\clspeed}{2}. \]

\begin{remark}[Free Boundary Problem]
Following the proof of Theorem~\ref{theorem:ca_evolution_in_standard_model}, it is easy to show that \eqref{eqn:evolution_equation_standard_model} also holds for a \emph{free boundary formulation}, where the Navier-Stokes equations are only solved in the liquid domain. The outer phase is represented just by a constant pressure field $p_0$ and the jump conditions \eqref{eqn:transmission_condition} are replaced by
\[ (p_0-p + S) \, \nsigma = \sigma \kappa \nsigma \quad \text{on} \quad \Sigma(t). \]
In particular, the viscous stress component $\inproduct{S \nsigma}{\tau}$ vanishes and Lemma~\ref{lemma:continuity_of_grad_v_at_gamma} is not required for the proof.
\end{remark}

\begin{corollary}
\label{corollary:stationarity}
Under the assumptions of Theorem~\ref{theorem:ca_evolution_in_standard_model}, a quasi-stationary solution, i.e. a solution with constant contact angle, satisfies
\begin{align}
\il \clspeed = 0 \quad \text{on} \quad \Gamma,
\end{align}
which means that either the contact line is at rest or $\il = 0$.
\end{corollary}
Consequently, a regular, non-trivial, quasi-stationary solution only exists if $\il$ vanishes at the contact line, i.e. in the free-slip case. On the other hand free-slip at the contact line implies that the contact angle is fixed for all regular solutions. This result confirms the observation from \cite{Schweizer.2001}, where it is stated that for a regular solution with $\il \in (0, \infty)$ and $\theta \equiv \pi/2$ ``the point of contact does not move''.

\begin{corollary}
Let $\il \geq 0$ and $(v,\gr\overline\Sigma)$ be a regular solution in the setting of Theorem~\ref{theorem:ca_evolution_in_standard_model} that satisfies the thermodynamic condition \eqref{eqn:advancing_receding_condition}. Then \eqref{eqn:evolution_equation_standard_model} implies
\begin{align*}
\dot{\theta} \geq 0 \quad \text{for} \quad \theta \geq \thetaeq \quad \text{and} \quad \dot{\theta} \leq 0 \quad \text{for} \quad \theta \leq \thetaeq.
\end{align*}
\end{corollary}
From this result, it follows that the system cannot evolve towards equilibrium with a regular solution in the setting of Theorem~\ref{theorem:ca_evolution_in_standard_model}.

\begin{corollary}
Let $\il \geq 0$ and $\{v,\gr\overline\Sigma\}$ be a regular classical solution of the PDE-system \eqref{eqn:pde-system-1}-\eqref{eqn:pde-system-4} in the setting of Theorem~\ref{theorem:ca_evolution_in_standard_model} that satisfies \eqref{eqn:advancing_receding_condition}. Let the initial condition be such that
\[ \theta(0,x) > \thetaeq \quad \forall \ x \in \Gamma(0), \]
where $\Gamma(0)=\partial\Sigma(0)$ is assumed to be bounded. Then it follows that
\[ \theta(t,x) \geq \min_{x' \in \Gamma(0)} \theta(0,x') > \thetaeq \]
for all $t \in I \cap [0,\infty)$ and $x \in \Gamma(t)$. That means that the system cannot relax to the equilibrium contact angle.
\end{corollary}
\begin{proof}
Consider $t \in I \cap [0,\infty)$ and $x_t \in \Gamma(t)$ arbitrary. Then there exists an $x_0 \in \Gamma(0)$ such that the unique solution $x(s)$ of the initial value problem
\begin{align}
\label{eqn:trajectory_ode}
x'(s) = v(s,x(s)), \ x(0)=x_0 \in \Gamma(0) 
\end{align}
satisfies $x(t)=x_t$. The point $x_0$ can be found by solving \eqref{eqn:trajectory_ode} backwards in time. By integration of \eqref{eqn:evolution_equation_standard_model} we conclude
\begin{align*}
\theta(t,x(t))&=\theta(t,x_t)=\theta(0,x(0)) + \int_0^t \frac{d}{ds} \, \theta(s,x(s)) \, ds \\
&= \theta(0,x_0) + \int_0^t \underbrace{\frac{\il \clspeed}{2}}_{\geq 0} (s,x(s)) \, ds \\
&\geq \theta(0,x_0) \geq \min_{x' \in \Gamma(0)} \theta(0,x') > \thetaeq.\qedhere
\end{align*}
\end{proof}

\subsection{Empirical contact angle models}
The literature contains a large variety of empirical contact angle models which prescribe the dynamic contact angle. For the simplest class of these models, it is assumed that $\theta$ can be described by a relation of the type\footnote{Note that there are also (numerical) models \cite{Afkhami.2009b}, based on the analysis of Cox \cite{Cox.1986}, which try to prescribe the \emph{apparent} contact angle rather than the \emph{actual} contact angle defined by \eqref{eqn:theta_definition}.}
\begin{align}
\theta = \ftheta(Ca, \thetaeq),
\end{align}
where $\thetaeq$ is the \emph{equilibrium} contact angle given by the Young equation \eqref{eqn:young}. The capillary number is defined as
\begin{align*}
Ca := \frac{\visc^- \, \clspeed}{\sigma}.
\end{align*}
Hence for a given system, $\ftheta$ is a function of the contact line velocity $\clspeed$, i.e.
\begin{align}
\label{eqn:empirical_model_v1}
\theta = \ftheta(\clspeed).
\end{align}
If this relation is invertible, one can also write ($\gtheta:=\ftheta^{-1}$)
\begin{align}
\label{eqn:empirical_model_v2}
\clspeed = \gtheta(\theta).
\end{align}
The following Corollary is an immediate consequence of this modeling.
\begin{corollary}
\label{corollary:evolution_equation_empirical_model}
Consider the model described in Theorem~\ref{theorem:ca_evolution_in_standard_model} together with the dynamic contact angle model \eqref{eqn:empirical_model_v1}. Let $\ftheta \in \mathcal{C}^1(\mathds{R})$. Then, for regular solutions in the sense of Theorem~\ref{theorem:ca_evolution_in_standard_model}, the contact line velocity obeys the evolution equation
\begin{align}
\label{eqn:ca_ode_1}
\ftheta'(\clspeed) \, \DDT{}\, \clspeed = \frac{\il\clspeed}{2}.
\end{align}
If the model from Theorem~\ref{theorem:ca_evolution_in_standard_model} is equipped with the contact angle model \eqref{eqn:empirical_model_v2} with $\gtheta \in \mathcal{C}^1(0,\pi)$, the contact angle for regular solutions in the sense of Theorem~\ref{theorem:ca_evolution_in_standard_model} follows the evolution equation
\begin{align}
\label{eqn:ca_ode_2}
\DDT{\theta} = \frac{\il \, \gtheta(\theta)}{2}. 
\end{align}
\end{corollary}

\begin{remark} From Corollary~\ref{corollary:evolution_equation_empirical_model} we draw the following conclusions.
 \begin{enumerate}[(i)]
  \item By adding on of the empirical models \eqref{eqn:empirical_model_v1} or \eqref{eqn:empirical_model_v2} to the model from Theorem~\ref{theorem:ca_evolution_in_standard_model} with a fixed slip length, the time evolution of $\theta$ and $\clspeed$ is, for regular solutions, already completely determined by the \emph{ordinary} differential equation \eqref{eqn:ca_ode_1} or \eqref{eqn:ca_ode_2}, respectively. But note that neither the momentum equation nor the normal part of the transmission condition involving the surface tension is used for its derivation. This means that, for regular solutions, neither external forces like gravity nor surface tension forces can influence the motion of the contact line.
  \item If the empirical function satisfies the thermodynamic condition \eqref{eqn:thermodynamic_conditions_f_g}, i.e.
  \[ \clspeed(\ftheta(\clspeed)-\thetaeq) \geq 0 \quad \text{or} \quad \gtheta(\theta)(\theta - \thetaeq) \geq 0, \]  
  respectively, there are only constant or \emph{monotonically increasing/decreasing} solutions for $\theta(t)$ (in Lagrangian coordinates).
  \item Moreover, we have the additional requirement that $\DT{\theta}=0$ for $\theta = \pi/2$, which dictates $\gtheta(\pi/2)=0=\gtheta(\thetaeq)$ (or $\il=0$ which means that $\theta$ is fixed).
 \end{enumerate}
\end{remark}

\subsection[Asymptotic solutions and regularity]{Asymptotic solutions and\\ regularity}
It is instructive to consider some examples of known asymptotic solutions to wetting flow problems. A classical example is the stationary two-dimensional Stokes problem in the free boundary formulation, i.e.\ the PDE system
\begin{align}
\visc \Delta v = \nabla p, \quad \nabla \cdot v = 0 \ \ &\text{in} \ \Omega\setminus\Sigma,\label{eqn:stokes-1} \\
v = 0 \ \ &\text{on} \ \Gamma,\label{eqn:stokes-2}\\
\inproduct{v}{\ndomega}=0 \ \ &\text{on} \ \partial\Omega\setminus\Gamma,\\ 
\lambda (\pdomega v-\vwall) + \pdomega S \ndomega = 0 \ \ &\text{on} \ \partial\Omega\setminus\Gamma, \label{eqn:stokes-3}\\
\inproduct{v}{\nsigma} = 0 \ \ &\text{on} \ \Sigma \label{eqn:stokes-4},\\
\psigma S \nsigma  = 0 \ \ &\text{on} \ \Sigma, \label{eqn:stokes-5}\\
p_0 - p + \inproduct{S \nsigma}{\nsigma}  = \sigma \kappa \ \ &\text{on} \ \Sigma, \label{eqn:stokes-6}
\end{align}
where $p_0$ is the constant outer pressure (see, e.g., \cite{Shikhmurzaev.2006}, \cite{Shikhmurzaev.2008}). Note that the equations are written in a frame of reference moving \emph{with} the contact line. So here we have a non-zero tangential wall velocity $\vwall$, which equals the contact line velocity. After introducing the scalar \emph{stream function} $\psi$ in polar coordinates $(r,\varphi)$, i.e.
\begin{align}
\label{eqn:stream_function}
v = v_r \hat{e}_r + v_\varphi \hat{e}_\varphi, \ v_r = \frac{1}{r} \, \partial_\varphi \psi, \ v_\varphi = - \partial_r \psi,
\end{align}
the incompressibility condition is automatically satisfied and the pressure can be eliminated from \eqref{eqn:stokes-1} leading to the \emph{biharmonic equation}
\begin{align}
\Delta^2 \psi = 0 \quad \text{in} \ \Omega\setminus\Sigma.
\end{align}
If necessary, the pressure can be recovered from $\psi$ via the relations (see \cite{Shikhmurzaev.2006}, \cite{Shikhmurzaev.2008} for details)
\begin{align}
\frac{\partial p}{\partial r} = \left(\frac{1}{r}\frac{\partial^3}{\partial r^2 \partial \varphi} + \frac{1}{r^3}\frac{\partial^3}{\partial \varphi^3} + \frac{1}{r}\frac{\partial^2}{\partial r \partial \varphi} \right) \psi,\label{eqn:pressure_relation_1}
\end{align}
and
\begin{align}
\frac{\partial p}{\partial \varphi} = -\left(r\frac{\partial^3}{\partial r^3} + \frac{\partial^2}{\partial r^2} + \frac{1}{r}\frac{\partial^3}{\partial r \partial \varphi^2} \right.\nonumber \\
\left. - \frac{1}{r}\frac{\partial}{\partial r} -\frac{2}{r^2}\frac{\partial^2}{\partial \varphi^2}\right) \psi.\label{eqn:pressure_relation_2}
\end{align}
As a first approximation, the system of equations \eqref{eqn:stokes-1}-\eqref{eqn:stokes-5} is solved on a wedge domain, i.e.\ for $0~<~r~<~\infty$ and $0~<~\varphi~<~\theta$. Afterwards, the normal stress condition \eqref{eqn:stokes-6} is evaluated to obtain a correction for the free surface shape. Rewriting the boundary conditions \eqref{eqn:stokes-2}-\eqref{eqn:stokes-5} in terms of the stream function leads to the PDE system
\begin{align}
\Delta^2 \psi = 0 \quad &\text{for} \quad r > 0, \ 0 < \varphi < \theta,\label{eqn:psi_stokes_1}\\
\psi = 0 \quad &\text{for} \quad r \geq 0, \ \varphi \in \{0, \theta\},\label{eqn:psi_stokes_2} \\
\frac{\partial^2 \psi}{\partial\varphi^2} = 0 \quad &\text{for} \quad r > 0, \ \varphi = \theta,\label{eqn:psi_stokes_3} 
\end{align}
together with the Navier boundary condition
\begin{align}
\frac{1}{r^2}\frac{\partial^2 \psi}{\partial\varphi^2} + \frac{1}{L} \left(\vwall - \frac{1}{r} \frac{\partial\psi}{\partial\varphi} \right) = 0 \quad &\text{for} \quad r > 0, \ \varphi = 0.\label{eqn:psi_stokes_4}
\end{align}
Moreover, the velocity field is required to be continuous up to the contact point. Since the frame of reference is moving with the contact line, a solution has to satisfy
\begin{align}
\label{eqn:reference_frame_condition}
\lim_{r \rightarrow 0} v_r = \lim_{r \rightarrow 0} \frac{1}{r} \, \partial_\varphi \psi = 0. 
\end{align}
Motivated by a separation of variables approach, one may consider \emph{special} solutions of the type
\begin{align}
\label{eqn:special_solutions}
\psi_\lambda(r,\varphi) = r^\lambda F_\lambda(\varphi).
\end{align}
A prominent example of such a solution is the one given by H. K. Moffatt \cite{Moffatt.1964}
\[
\psi_1(r,\varphi) = \frac{r[(\varphi-\theta)\sin\varphi - \varphi \sin(\varphi-\theta)\cos\theta]}{\sin \theta \cos \theta - \theta},
\]
which satisfies a no-slip condition on $\partial\Omega\setminus\Gamma(t)$ (for $\vwall=1$). However, the resulting velocity field is discontinuous at the point of contact and the pressure diverges proportional to $1/r$, which makes it impossible to satisfy the normal stress condition \eqref{eqn:stokes-6} (see \cite{Shikhmurzaev.2006},\cite{Shikhmurzaev.2008} for details).  In fact, it can be easily seen from \eqref{eqn:stream_function} that for regular $F_\lambda \not\equiv 0$, the velocity is continuous on $\overline{\Omega}$ if $\lambda > 1$. Then condition \eqref{eqn:reference_frame_condition} is also satisfied by $\psi_\lambda$.\newline
\newline
Note that for $\lambda < 2$ the stress is unbounded for $r \rightarrow 0$ and Theorem \ref{theorem:ca_evolution_in_standard_model} does \emph{not} apply. For $\lambda > 2$ we observe that the tangential stress component vanishes at the contact line, i.e.
\[ \frac{1}{r^2} \frac{\partial^2 \psi_\lambda}{\partial \varphi^2} (0,0) = 0. \]
Using equation \eqref{eqn:psi_stokes_4} this implies that either the contact line is at rest or $L \rightarrow \infty$, in agreement with Theorem~\ref{theorem:ca_evolution_in_standard_model}. It follows from the equations \eqref{eqn:pressure_relation_1} and \eqref{eqn:pressure_relation_2} that $\lambda > 2$ is also a sufficient condition to obtain a \emph{finite} pressure at the moving contact line.
\begin{figure}[ht]
\centering
\includegraphics[width=0.6\columnwidth]{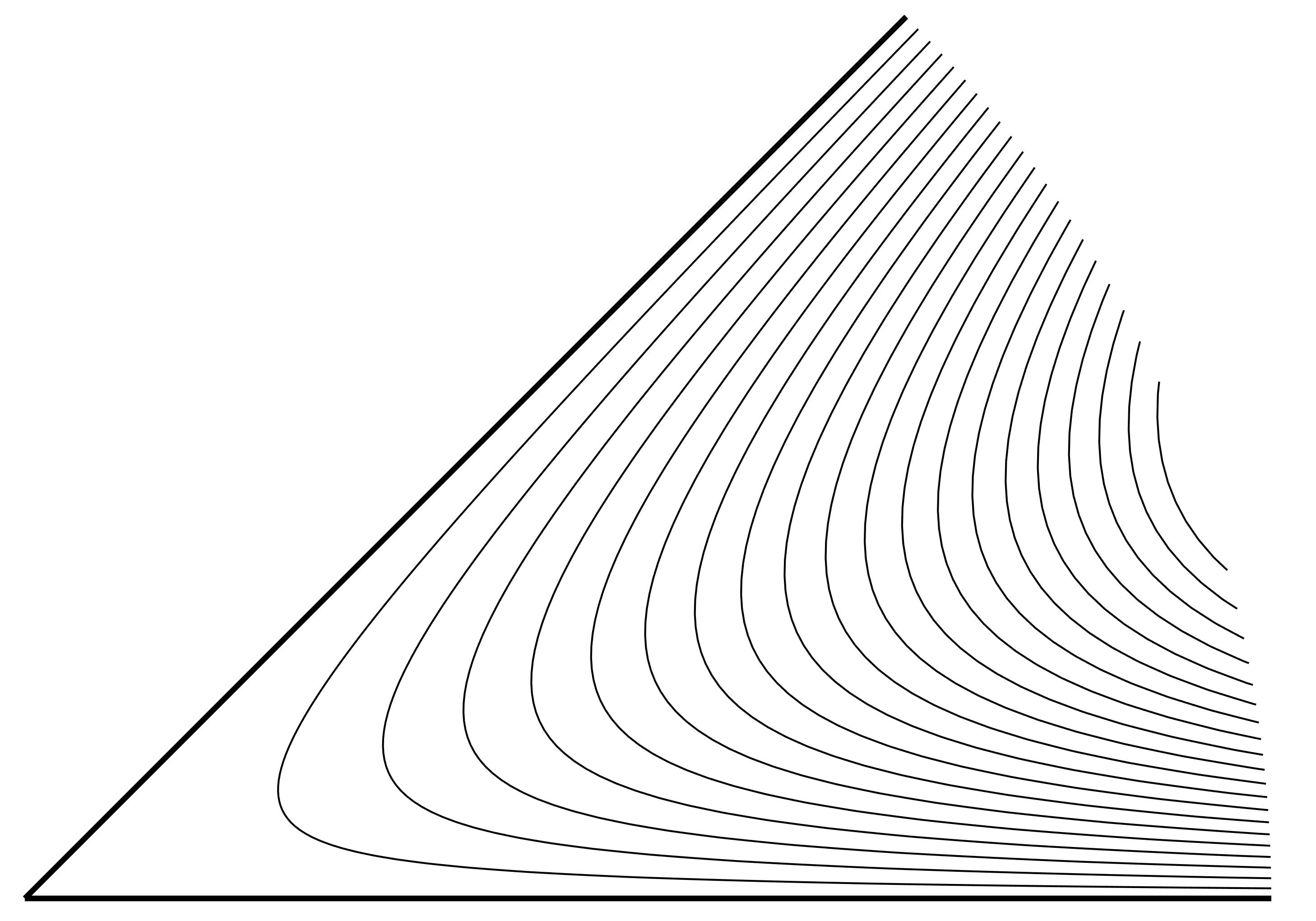}
\caption{Streamlines for the field given by \eqref{eqn:borderline_example}.}
\label{fig:contour_borderline_case}
\end{figure}

It is interesting to take a look at the borderline case $\lambda~=~2$. For example, the stream function (see \cite{Shikhmurzaev.2006} and Figure~\ref{fig:contour_borderline_case})
 \begin{align}
 \label{eqn:borderline_example}
 \psi_2(r,\varphi) = r^2 \, \frac{\vwall}{L_0} \left( - \frac{1}{4} + \frac{\varphi}{\pi} + \frac{1}{4} \cos(2\varphi) \right),
 \end{align}
 describes a velocity field which is an exact solution of \eqref{eqn:psi_stokes_1} - \eqref{eqn:reference_frame_condition} with $\theta = \pi/4$ and a finite slip length $L_0$ at the contact point. However, the velocity field is \emph{not} differentiable at the contact line and the pressure is logarithmically singular. In order to construct a solution to the free boundary problem, this requires a correction to the free surface with a singular curvature at the contact point. Due to the lack of differentiability, Theorem~\ref{theorem:ca_evolution_in_standard_model} also does \emph{not} apply in this case.\newline
 \newline
 But the case $\lambda~=~2$ also includes examples for stream functions, where the required regularity is met. In fact, it can be shown that these are precisely given by
 \[ \psi_2(r,\varphi) = r^2 (c_1 + c_2 \sin \varphi \cos \varphi + c_3 \sin^2 \varphi), \]
 where $c_1,c_2,c_3 \in \RR$. This class of stream functions represents the (three-dimensional) space of \emph{linear} divergence free velocity fields in two spatial dimensions satisfying $v=0$ at $r=0$. The impermeability condition $v_\varphi = 0$ for $\varphi=0$ implies $c_1=0$. The Navier condition \eqref{eqn:psi_stokes_4} yields
 \begin{align*}
 c_3 = - \frac{\vwall}{2 L_0},
 \end{align*}
 where $L_0$ denotes the slip length at the contact point. Finally, the tangential stress condition \eqref{eqn:psi_stokes_3} allows to determine the constant $c_2$. For $\theta \neq \pi/2$ we obtain the stream function
 \begin{align*}
 \psi_2(r,\varphi) = -\frac{\vwall r^2}{2 L_0}\left(\cot(2\theta) \sin \varphi \cos \varphi +\sin^2 \varphi \right).
 \end{align*}
 The corresponding time derivative of the contact angle is (as expected)
 \[ \DDT{\theta} = \frac{\vwall}{2L_0} = \frac{\clspeed}{2L_0}. \]
 Clearly, this is \emph{not} a quasi-stationary solution since \eqref{eqn:stokes-4} is not satisfied. As already pointed out in Corollary~\ref{corollary:stationarity}, a sufficiently regular, non-trivial \emph{quasi-stationary} solution only exists in the free-slip case.

\section{Remarks on more general models}
\label{section:generalizations}
\subsection{Marangoni effect}
An obvious generalization of the model described in Theorem~\ref{theorem:ca_evolution_in_standard_model} is to include the effect of non-constant fluid-fluid surface tension. In this case, the interfacial transmission condition for the stress reads as
\[ \jump{p \mathds{1} - S} \nsigma = \sigma \kappa \nsigma + \nablasigma \sigma. \]

\begin{theorem}[]
\label{theorem:ca_evolution_marangoni}
Let $\Omega\subset\RR^3$ (or $\Omega\subset\RR^2$) be a half-space with boundary $\partial\Omega$, $\visc^\pm >0$, $\jump{\visc}\neq 0$, $\il \in \mathcal{C}(\gr\partial\Omega)$ and $(v,\gr\overline\Sigma)$ with $v \in \velspace$, $\gr\overline\Sigma$ a $\mathcal{C}^{1,2}$-family of moving hypersurfaces with boundary, be a classical solution of the PDE-system
\begin{align*}
\nabla \cdot v = 0 \quad &\text{in} \ \Omega\setminus\Sigma(t), \\
\jump{v} = 0, \quad \mathcal{P}_\Sigma \jump{-S} \, \nsigma = \nabla_\Sigma \sigma \quad &\text{on} \ \Sigma(t),\\
\left\langle v, \ndomega \right\rangle = 0, \ \il \pdomega v + 2 \pdomega D \ndomega = 0 \quad &\text{on} \ \partial\Omega\setminus\Gamma(t),\\
\normalspeed = \inproduct{v}{\nsigma} \quad &\text{on} \ \Sigma(t)
\end{align*}
with $\theta \in (0,\pi)$ on $\gr\Gamma$. Then the evolution of the contact angle is given by
\begin{align}
\label{eqn:ca_evolution_marangoni}
\frac{D \theta}{D t} = \frac{1}{2}\left(\il \clspeed - \frac{\partial_\tau \sigma}{\jump{\visc}} \right).
\end{align}
\end{theorem}
\begin{proof}
The proof is analogous to the proof of Theorem~\ref{theorem:ca_evolution_in_standard_model}. We proceed as follows:
While \eqref{eqn:3d_evolution_equation} is still valid, the jump condition for the viscous is replaced by
 \[ \inproduct{\jump{S}\nsigma}{\tau} = - \partial_\tau \sigma. \]
Using again Lemma~\ref{lemma:continuity_of_grad_v_at_gamma} we find
\begin{align}
\label{eqn:3d_evolution_intermediate_equation_marangoni}
  -\frac{\partial_\tau \sigma}{2\jump{\visc}} &= \sin \theta \cos \theta \left(- \inproduct{\nabla v \, \ngamma}{\ngamma} + \inproduct{\nabla v \, \ndomega}{\ndomega} \right) \nonumber\\
    & + (\sin^2 \theta - \cos^2 \theta) \frac{\il \clspeed}{2} \quad \text{on} \ \Gamma.
\end{align}
The claim follows by inserting \eqref{eqn:3d_evolution_intermediate_equation_marangoni} into the contact angle evolution equation \eqref{eqn:3d_evolution_equation}.
\end{proof}
This result shows that in this case regular solutions with advancing contact line and $\dot{\theta} < 0$ are possible. To obtain a non-trivial quasi-stationary state, a surface tension gradient
\begin{align}
\label{eqn:surface_tension_gradient}
\partial_\tau \sigma = \jump{\visc} a \clspeed 
\end{align}
has to be present at the contact line.

\begin{remark}
Figure~\ref{fig:velocity_fields_constant_theta} shows an example of a linear velocity field in two spatial dimensions with $\theta=\pi/4$ and a gradient in surface tension corresponding to \eqref{eqn:surface_tension_gradient}. The field satisfies Navier slip with $L>0$ and is plotted in a \emph{co-moving} reference frame. The streamlines are tangent to the interface and the contact angle does not change.\newline
\newline
This situation is not possible for the case of \emph{constant} surface tension, visualized in Figure~\ref{fig:velocity_fields_constant_sigma}. In the case $\theta=\pi/4$ and $\partial_\tau \sigma = 0$ equation \eqref{eqn:3d_evolution_intermediate_equation_marangoni} together with the incompressibility condition
\[ 0 = \inproduct{\nabla v \, \ngamma}{\ngamma} + \inproduct{\nabla v \, \ndomega}{\ndomega}  \]
implies 
\[ \inproduct{\nabla v \, \ngamma}{\ngamma}=\inproduct{\nabla v \, \ndomega}{\ndomega} = 0. \]
Therefore, the linear part of the velocity field has a quite simple form. In the reference frame of the solid wall, it is given as
\[ (u,v)(x,y) = \clspeed \left(1 + \frac{y}{L}, 0 \right). \]
Figure~\ref{fig:velocity_fields_constant_sigma} shows the field in a co-moving reference frame. Clearly, the field geometry leads to an \emph{increase} in the contact angle (clockwise rotation in this example).
\end{remark}
\begin{figure*}[hbt]
\label{fig:velocity_fields}
\subfigure[Constant contact angle, $\nablasigma \sigma \neq 0$.]{\includegraphics[width=\columnwidth]{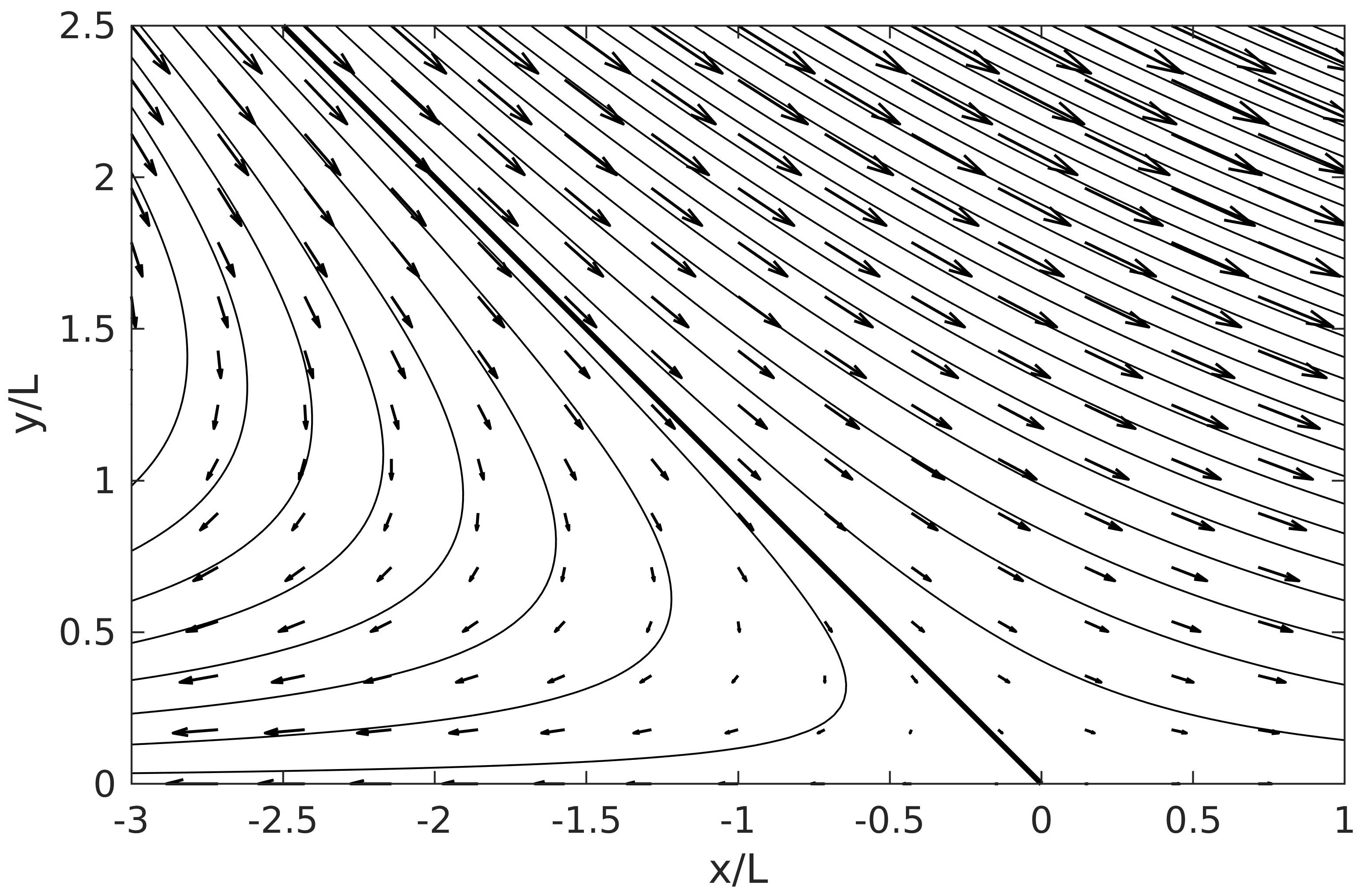}\label{fig:velocity_fields_constant_theta}}
\subfigure[Constant surface tension, $\dot{\theta} > 0$.]{\includegraphics[width=\columnwidth]{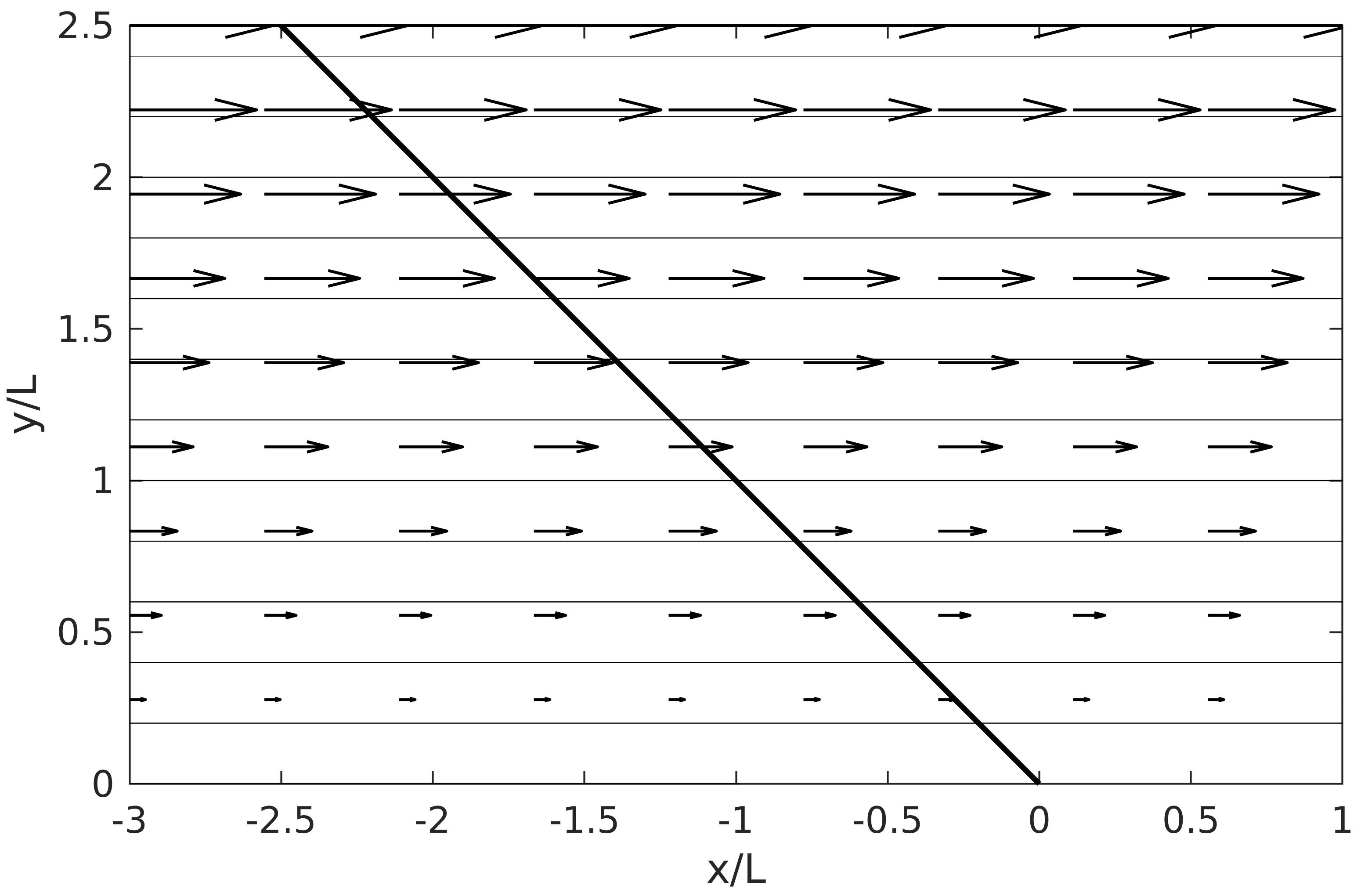}\label{fig:velocity_fields_constant_sigma}}
\caption{Linear velocity fields satisfying Navier slip with $L>0$ in a co-moving reference frame.}
\end{figure*}

\subsection{Interfacial slip}
Another possible generalization of the model is to allow for slip at the fluid-fluid interface. In this case, one only requires continuity of the \emph{normal} component of the fluid velocity, i.e.
\[ \inproduct{\jump{v}}{\nsigma} = 0 \quad \text{on} \ \Sigma(t),\]
which means that there is no mass flux from one phase to the other. To describe the evolution of the interface, one can use both of the fluid velocities $v^\pm$ in the kinematic conditions
\[ \normalspeed = \inproduct{v^\pm}{\nsigma}, \quad \clspeed = \inproduct{v^\pm}{\ngamma}. \]
This gives rise to two distinct Lagrangian derivative operators. To formulate the following Theorem, we choose the interfacial velocity field
\begin{align}
\label{eqn:interfacial_slip_vsigma}
\vsigma := \frac{\visc^+ v^+ - \visc^- v^-}{\visc^+ - \visc^-} = \frac{\jump{\visc v}}{\jump{\visc}}.
\end{align}
Clearly, $\vsigma$ also satisfies the above mentioned kinematic conditions. Hence we can define a Lagrangian time derivative according to $\vsigma$. 

\begin{lemma}
The Lagrangian derivatives with respect to $v^+$, $v^-$ and $\vsigma$ satisfy the relation
\begin{align}
\label{eqn:ddt_relation}
(\visc^+ - \visc^-)\, \DSigmaDT{} = \visc^+ \DDTp{} - \visc^- \DDTm{}. 
\end{align}
\end{lemma}
\begin{proof}
Each Lagrangian derivative along the contact line may be decomposed as 
\[ \DDTpm{} = \partial^\Gamma_t + v^\pm_\parallel \cdot \nablagamma, \]
where $\partial^\Gamma_t$ is the ``contact line Thomas derivative'' following the normal motion of the contact line, $\nablagamma$ is the gradient along $\Gamma$ and $v^\pm_\parallel$ is the component of $v^\pm$ tangential to the contact line. Multiplication with $\visc^\pm$ yields
\begin{align*}
\visc^+ \DDTp{} - \visc^- \DDTm{} &= \jump{\visc} \partial^\Gamma_t + (\visc^+ v^+_\parallel - \visc^- v^-_\parallel) \cdot \nablagamma \\
&= \jump{\visc} \left(\partial^\Gamma_t + \frac{\visc^+ v^+_\parallel - \visc^- v^-_\parallel}{\visc^+ - \visc^-} \cdot \nablagamma \right)\\
&= \jump{\visc} \DSigmaDT{}.\qedhere
\end{align*}
\end{proof}

Note that Lemma~\ref{lemma:continuity_of_grad_v_at_gamma} cannot be used for the proof of Theorem~\ref{theorem:interfacial_slip} since $v$ is no longer assumed to be continuous. Moreover, the following statement does \emph{not} require incompressibility of the flow.
\begin{theorem}[]
\label{theorem:interfacial_slip}
Let $\Omega\subset\RR^3$ be a half-space with boundary $\partial\Omega$, $\visc^\pm >0$, $\jump{\visc}\neq 0$ and $(v,\gr\overline\Sigma)$ with 
\[ v \in \mathcal{C}^1(\gr \overline{\Omega^+}) \cap \mathcal{C}^1(\gr \overline{\Omega^-})  \]
and $\gr\overline\Sigma$ a $\mathcal{C}^{1,2}$-family of moving hypersurfaces with boundary, be a classical solution of the PDE-system
\begin{align*}
\jump{\inproduct{v}{\nsigma}} = 0, \quad \mathcal{P}_\Sigma \jump{-S} \, \nsigma = \nablasigma \sigma \quad &\text{on} \ \Sigma(t),\\
\inproduct{v}{\ndomega}=0, \ \lambda \pdomega v + \pdomega S \ndomega = 0 \quad &\text{on} \ \partial\Omega\setminus\Gamma(t),\\
\normalspeed = \inproduct{v^\pm}{\nsigma} \quad &\text{on} \ \Sigma(t).
\end{align*}
Moreover, let 
\[ \lambda \in \mathcal{C}(\gr\overline{\partial\Omega^+})\cap\mathcal{C}(\gr\overline{\partial\Omega^-}) \]
and $\theta \in (0,\pi)$ on $\gr\Gamma$. Then the evolution of the contact angle is given by
\begin{align}
\label{eqn:ca_evolution_with_if_slip}
\DSigmaDT{\theta} = \frac{1}{2} \left( \clspeed \frac{\jump{\lambda}}{\jump{\visc}} - \frac{\partial_\tau \sigma}{\jump{\visc}} \right),
\end{align}
where $\DSigmaDT{}$ is the Lagrangian time-derivative according to the surface velocity field \eqref{eqn:interfacial_slip_vsigma}.
\end{theorem}
\begin{proof}
We start from \eqref{eqn:main_result} and, by a change of coordinates, obtain
 \begin{align*}
 \DDTpm{\theta} = &- \sin^2 \theta \inproduct{\nabla v^\pm \, \ndomega}{\ngamma}\\
 &+ \sin \theta \cos \theta (\inproduct{\nabla v^\pm \, \ndomega}{\ndomega} - \inproduct{\nabla v^\pm \, \ngamma}{\ngamma}),
 \end{align*}
 where the two time derivatives may now be different. Using the Navier condition, we can replace the first term according to
 \begin{align*}
  \DDTpm{\theta} = &\sin^2 \theta \, a^\pm \clspeed\\
  &+ \sin \theta \cos \theta (\inproduct{\nabla v^\pm \, \ndomega}{\ndomega} - \inproduct{\nabla v^\pm \, \ngamma}{\ngamma}).
 \end{align*}
 Note that $\il$ may now be discontinuous at $\Gamma$. The next step is to express $\dot{\theta}$ by means of the jump in normal stress. Thanks to $\visc^\pm > 0$ we can write
 \begin{align*} 
 \inproduct{\nabla v^\pm \, \ngamma}{\ngamma} &= \frac{1}{2\visc^\pm} \inproduct{S^\pm \ngamma}{\ngamma},\\
 \inproduct{\nabla v^\pm \, \ndomega}{\ndomega} &= \frac{1}{2\visc^\pm} \inproduct{S^\pm \ndomega}{\ndomega}. 
 \end{align*}
 Plugging this into the above equation for $\dot{\theta}$ gives
 \begin{align}
 \label{eqn:interfacial_slip_eq_4}
 \DDTpm{\theta} = &\sin^2 \theta \, \il^\pm \clspeed\nonumber\\
 &+ \frac{\sin \theta \cos \theta}{2\visc^\pm}  (\inproduct{S^\pm \ndomega}{\ndomega}-\inproduct{S^\pm \ngamma}{\ngamma}).
 \end{align}
 We introduce $\jump{S}$ by adding a zero term, i.e.
 \begin{align}
 &\DDTp{\theta} = \sin^2 \theta \, \il^+ \clspeed + \frac{\sin \theta \cos \theta}{2\visc^+} (\inproduct{(S^+-S^-) \ndomega}{\ndomega}\nonumber \\
  &-\inproduct{(S^+-S^-) \ngamma}{\ngamma}+\inproduct{S^- \ndomega}{\ndomega} -  \inproduct{S^- \ngamma}{\ngamma}).\label{eqn:interfacial_slip_eq_1}
 \end{align}
 Using the second version of the equation \eqref{eqn:interfacial_slip_eq_4} we have
 \begin{align*} 
 &\visc^- \DDTm{\theta} - \sin^2 \theta \, \il^- \visc^- \clspeed \\
 = \quad &\frac{\sin \theta \cos \theta}{2} \left(\inproduct{S^- \ndomega}{\ndomega} - \inproduct{S^- \ngamma}{\ngamma} \right).
 \end{align*}
 Together with \eqref{eqn:interfacial_slip_eq_1} we obtain (using $\lambda^\pm = \il^\pm \visc^\pm$)
 \begin{align}
 \label{eqn:interfacial_slip_eq_2}
 \visc^+ \DDTp{\theta} - \visc^- \DDTm{\theta} = \sin^2 \theta \jump{\lambda} \clspeed\nonumber\\
 +\frac{\sin \theta \cos \theta}{2} (\inproduct{\jump{S} \ndomega}{\ndomega} - \inproduct{\jump{S} \ngamma}{\ngamma}).
 \end{align}
Now we exploit the validity of both the Navier and the jump condition for the stress at the contact line. From $\psigma \jump{S} \nsigma = - \nablasigma\sigma$ we obtain
 \begin{align}
  -\partial_\tau \sigma = \inproduct{\jump{S}\nsigma}{\tau} = (\cos^2 \theta - \sin^2 \theta) \inproduct{\jump{S}\ndomega}{\ngamma}\nonumber\\
  + \sin \theta \cos \theta (- \inproduct{\jump{S}\ngamma}{\ngamma}+\inproduct{\jump{S}\ndomega}{\ndomega})\label{eqn:interfacial_slip_eq_3}
 \end{align}
From the Navier condition, i.e.
 \[ \lambda^\pm \clspeed + \inproduct{S^\pm \ndomega}{\ngamma} = 0,  \]
 we infer
 \[ \inproduct{\jump{S}\ndomega}{\ngamma} = - \jump{\lambda} \clspeed \]
 by taking the trace. Combined with \eqref{eqn:interfacial_slip_eq_3} we obtain
 \begin{align*} 
 &\sin \theta \cos \theta (\inproduct{\jump{S}\ndomega}{\ndomega} - \inproduct{\jump{S}\ngamma}{\ngamma})\\
 = \quad &(\cos^2 \theta - \sin^2 \theta) \jump{\lambda} \clspeed - \partial_\tau \sigma. 
 \end{align*}
Plugging in this expression into \eqref{eqn:interfacial_slip_eq_2}, we arrive at
 \begin{align*}
 & \quad \visc^+ \DDTp{\theta} - \visc^- \DDTm{\theta} \\
 &=  \jump{\lambda} \clspeed \left(\sin^2 \theta + \frac{\cos^2 \theta - \sin^2 \theta}{2}\right) - \frac{\partial_\tau \sigma}{2} \\
 &= \frac{1}{2}(\jump{\lambda}\clspeed - \partial_\tau \sigma).
 \end{align*}
Now the claim follows from \eqref{eqn:ddt_relation}.
\end{proof}

\begin{remark}
\begin{enumerate}[(i)]
 \item If the flow is incompressible and the velocity is continuous across $\Sigma$, equation~\eqref{eqn:jump_condition_slip_length} implies that the slip length has to be continuous across the contact line to allow for non-trivial regular solutions. In this case, we have $\jump{\lambda} = \il \jump{\visc}$ and \eqref{eqn:ca_evolution_with_if_slip} reduces to \eqref{eqn:ca_evolution_marangoni}.
 \item If the surface tension $\sigma$ is constant, we obtain the evolution equation 
 \[ \DSigmaDT{\theta} = \frac{\clspeed}{2} \frac{\jump{\lambda}}{\jump{\visc}} \]
 If both jumps have the same sign, i.e. if
 \[ \frac{\jump{\lambda}}{\jump{\visc}} \geq 0,\]
 the qualitative behavior of regular solutions is still the same as in Theorem~\ref{theorem:ca_evolution_in_standard_model}.
\end{enumerate}
\end{remark}

\subsection{Systems with phase change}
So far we only discussed the case, when no phase transitions occur. We now generalize the results for non-zero mass transfer across the fluid-fluid interface. Given an interface with interface normal field $\nsigma$ and normal velocity $\normalspeed$, the one-sided \emph{mass transfer fluxes} are defined as
\begin{align*}
\dot{m}^\pm = \rho^\pm (v^\pm \cdot \nsigma - \normalspeed) \quad \text{on} \quad \gr\overline\Sigma.
\end{align*}
If the interface is not able to \emph{store} mass, the mass transfer flux has to be continuous, i.e. 
\begin{align}
\label{eqn:consistency_condition_mass_transfer}
\jump{\dot{m}}=0 \quad \Leftrightarrow \quad \jump{\rho v} \cdot \nsigma = \jump{\rho} \normalspeed.
\end{align}
Note that models for dynamic wetting allowing for mass on the fluid-fluid interface have also been considered under the name Interface Formation Model, see~\cite{Shikhmurzaev.1993},\cite{Shikhmurzaev.2008}. In the case without interfacial mass, the interfacial normal velocity can be expressed as
\[ \normalspeed = v^\pm \cdot \nsigma - \frac{\dot{m}}{\rho^\pm}.  \]
Since the interface is now transported by $\normalspeed \neq v^\pm \cdot \nsigma$, the mass flux influences the evolution of the interface. From the above relation, it follows that the mass transfer flux is related to the jump in the normal component of $v$ according to
\begin{align}
\label{eqn:normal_velocity_jump}
\jump{v} \cdot \nsigma = \jump{1/\rho} \dot{m} \quad \text{on} \quad \gr\Sigma.
\end{align}
At the contact line, we can express $\nsigma$ via $\ngamma$ and $\ndomega$, i.e.
\[ \jump{v} \cdot \nsigma = \sin \theta \jump{v} \cdot \ngamma - \cos \theta \jump{v} \cdot \ndomega \quad \text{on} \ \gr\Gamma. \]
For simplicity, we consider in the following the case of two spatial dimensions\footnote{Note that equation \eqref{eqn:consistency_condition_mass_transfer} implies that the ``natural'' interfacial velocity to be considered here is \eqref{eqn:def_vsigma_mass_transfer}, i.e.\ the bulk velocities should be weighted with the density $\rho$. Recall that in the case of interfacial slip without mass transfer, the natural interfacial velocity to choose is \eqref{eqn:interfacial_slip_vsigma}, i.e.\ weighted with the viscosity $\visc$. In the present case of interfacial slip with mass transfer, it is not obvious which interfacial velocity to choose. We exclude this problem by restricting the Theorem to the 2D case, where the interfacial velocity following the contact line is unique.}. If we assume $v^\pm$ to satisfy the impermeability condition, we find the following relation for the jump of $v$ at the contact line
\begin{align}
\label{eqn:jump_of_v_at_gamma}
\sin \theta \jump{v}_{|\Gamma} = \jump{1/\rho} \dot{m} \, \ngamma.
\end{align}
Note that the above relation only holds in two spatial dimensions. A slip tangential to the contact line may be present in three dimensions. In this case one can only state that
\[ (\mathds{1} - \inproduct{\tgamma}{\cdot} \tgamma) \sin \theta \jump{v}_{|\Gamma} = \jump{1/\rho} \dot{m} \, \ngamma.  \]

\begin{theorem}[]
\label{theorem:mass_transfer}
Let $\Omega\subset\RR^2$ be a half-space with boundary $\partial\Omega$, $\visc^\pm >0$, $\jump{\visc}\neq 0$ and $(v,\gr\overline\Sigma)$ with 
\[ v \in \mathcal{C}^1(\gr \overline{\Omega^+}) \cap \mathcal{C}^1(\gr \overline{\Omega^-})  \]
and $\gr\overline\Sigma$ a $\mathcal{C}^{1,2}$-family of moving hypersurfaces with boundary, be a classical solution of the PDE-system
\begin{align*}
\dot{m} \, \psigma \jump{v} + \psigma \jump{-S} \, \nsigma = \nablasigma \sigma \quad &\text{on} \ \Sigma(t),\\
\inproduct{v}{\ndomega}=0, \ \lambda \pdomega v + \pdomega S \ndomega = 0 \quad &\text{on} \ \partial\Omega\setminus\Gamma(t),\\
\normalspeed = \inproduct{v^\pm}{\nsigma} - \frac{\dot{m}}{\rho^\pm} \quad &\text{on} \ \Sigma(t)
\end{align*}
with $\theta \in (0,\pi)$ on $\gr\Gamma$. Then, the interfacial velocity field $\vsigma \in \mathcal{C}^1(\gr\overline\Sigma)$ defined as
\begin{align}
\label{eqn:def_vsigma_mass_transfer}
\vsigma := \frac{\jump{\rho v}}{\jump{\rho}} = \frac{\rho^+ v^+ - \rho^- v^-}{\rho^+ - \rho^-} \quad \text{on} \ \gr\overline\Sigma
\end{align}
satisfies the consistency conditions
\begin{align} 
\normalspeed = \inproduct{\vsigma}{\nsigma} \quad &\text{on} \ \gr\Sigma\label{eqn:kinematic_mass_transfer_1},\\
\clspeed = \inproduct{\vsigma}{\ngamma} \quad &\text{on} \ \gr\Gamma\label{eqn:kinematic_mass_transfer_2}
\end{align}
and the corresponding evolution of the contact angle is given by
\begin{align}
\label{eqn:contact_angle_evolution_mass_transfer_1}
\jump{\visc} \DSigmaDT{\theta} = \frac{\jump{\lambda}\clspeed}{2} - \frac{\partial_\tau \sigma}{2} - \jump{\visc/\rho} \partial_\tau \dot{m} \nonumber\\
+ \, \dot{m} \, (- \jump{\visc \partial_\tau (1/\rho)} - \kappa \cot \theta \jump{\frac{\visc}{\rho}}\nonumber\\
- \frac{\cot\theta}{2} \jump{\frac{1}{\rho}} \dot{m} + \frac{1}{2\sin\theta}\jump{\frac{\lambda}{\rho}} ).
\end{align}
In the special case of zero mass flux at the contact line, the above equation simplifies to
\begin{align}
\label{eqn:contact_angle_evolution_mass_transfer_2}
\jump{\visc} \DSigmaDT{\theta} = \frac{\jump{\lambda}\clspeed}{2} - \frac{\partial_\tau \sigma}{2} - \jump{\visc/\rho} \partial_\tau \dot{m}.
\end{align}
\end{theorem}
\begin{proof}
We observe that, as a consequence of \eqref{eqn:consistency_condition_mass_transfer}, the velocity $\vsigma$ defined by \eqref{eqn:def_vsigma_mass_transfer} satisfies the consistency condition \eqref{eqn:kinematic_mass_transfer_1}. It also satisfies $\vsigma \cdot \ndomega = 0$ on $\gr\Gamma$ since $v^\pm$ are tangential to $\partial\Omega$. Hence, Lemma~\ref{lemma:kinematic_conditions} implies that  \eqref{eqn:kinematic_mass_transfer_2} also holds. Moreover, it is easy to check that $\vsigma$  can be expressed in two different ways as
\begin{align}
\vsigma = v^\pm - \frac{\jump{v}}{\rho^\pm \jump{1/\rho}} \quad \text{on} \quad \gr\overline\Sigma.
\end{align}
Applying Theorem~\ref{theorem:contact_angle_evolution} yields, using \eqref{eqn:normal_velocity_jump},
\begin{align}
\label{eqn:dot_theta_mass_transfer_1}
\DSigmaDT{\theta} &= \inproduct{\partial_\tau v^\pm}{\nsigma} - \inproduct{\frac{\partial}{\partial \tau} \frac{\jump{v}}{\rho^\pm \jump{1/\rho}}}{\nsigma}\nonumber\\
&= \inproduct{\partial_\tau v^\pm}{\nsigma} - \partial_\tau \inproduct{\frac{\jump{v}}{\rho^\pm \jump{1/\rho}}}{\nsigma}\nonumber\\
& - \frac{\jump{v}\cdot\tau}{\rho^\pm \jump{1/\rho}} \inproduct{\tau}{\partial_\tau \nsigma}\nonumber\\
&= \inproduct{\nabla v^\pm \tau}{\nsigma} - \partial_\tau \frac{\dot{m}}{\rho^\pm} - \kappa \cot \theta \frac{\dot{m}}{\rho^\pm}\nonumber\\
&=: \inproduct{\nabla v^\pm \tau}{\nsigma} + \mathcal{R}.
\end{align}
Here we used the relation 
\[ \frac{\jump{v}\cdot\tau}{\jump{1/\rho}} = \frac{\dot{m}}{\sin\theta} \, \ngamma \cdot \tau = - \cot \theta \, \dot{m} \quad \text{on} \ \gr\Gamma, \]
which follows from \eqref{eqn:jump_of_v_at_gamma}. Multiplication of the first term in \eqref{eqn:dot_theta_mass_transfer_1} with $\visc^\pm$ together with a change of basis vectors yields 
\begin{align*}
\visc^\pm \inproduct{\nabla v^\pm \tau}{\nsigma} = - \sin^2 \theta \, \visc^\pm \inproduct{\nabla v^\pm \ndomega}{\ngamma} \\
+ \sin \theta \cos \theta \, \visc^\pm (\inproduct{\nabla v^\pm \ndomega}{\ndomega}  - \inproduct{\nabla v^\pm \ngamma}{\ngamma})\\
= - \sin^2 \theta \inproduct{S^\pm \ndomega}{\ngamma} + \frac{\sin \theta \cos \theta}{2} \\
(\inproduct{S^\pm \ndomega}{\ndomega} - \inproduct{S^\pm \ngamma}{\ngamma}).
\end{align*}
We may now multiply \eqref{eqn:dot_theta_mass_transfer_1} by $\visc^\pm$ to find
\begin{align}
\label{eqn:dot_theta_mass_transfer_2}
\jump{\visc} \DSigmaDT{\theta} = \jump{\visc \mathcal{R}} - \sin^2 \theta \inproduct{\jump{S}\ndomega}{\ngamma}\nonumber\\
+ \frac{\sin \theta \cos \theta}{2}(\inproduct{\jump{S}\ndomega}{\ndomega}-\inproduct{\jump{S}\ngamma}{\ngamma}).
\end{align}
The tangential stress condition at $\Gamma$ reads
\begin{align*}
\inproduct{\jump{S}\nsigma}{\tau} &= - \partial_\tau \sigma + \dot{m} \jump{v} \cdot \tau \\
&= - \partial_\tau \sigma - \cot \theta \jump{1/\rho} \dot{m}^2.
\end{align*}
We rewrite the left-hand side using the expansions for $\nsigma$ and $\tau$, i.e.
\begin{equation}
\label{eqn:stress_condition_mass_transfer}
\begin{aligned}
\inproduct{\jump{S}\nsigma}{\tau} = (\cos^2 \theta - \sin^2 \theta) \inproduct{\jump{S}\ndomega}{\ngamma} \\
+ \sin \theta \cos \theta (\inproduct{\jump{S}\ndomega}{\ndomega}-\inproduct{\jump{S}\ngamma}{\ngamma})\\
= - \partial_\tau \sigma - \cot \theta \jump{1/\rho} \dot{m}^2.
\end{aligned}
\end{equation}
From \eqref{eqn:dot_theta_mass_transfer_2} and \eqref{eqn:stress_condition_mass_transfer} we obtain
\begin{align}
\jump{\visc} \DSigmaDT{\theta} = \jump{\visc\mathcal{R}} -\frac{1}{2} \inproduct{\jump{S}\ndomega}{\ngamma}\nonumber\\
- \frac{\partial_\tau \sigma}{2} - \frac{\cot\theta}{2}\jump{1/\rho} \dot{m}^2.
\end{align}
Using \eqref{eqn:vsigma_vgamma} we can compute the contact line velocity
\begin{align*} 
\sin \theta \, \clspeed &= v^\pm \cdot \nsigma - \frac{\jump{v}\cdot\nsigma}{\rho^\pm \jump{1/\rho}}\\
&= \sin \theta \, v^\pm \cdot \ngamma - \frac{\dot{m}}{\rho^\pm}.
\end{align*}
Hence the contact line velocity reads
\begin{align}
\label{eqn:cl_speed_mass_transfer}
\clspeed = v^\pm \cdot \ngamma - \frac{\dot{m}}{\sin\theta \rho^\pm} \quad \text{on} \quad \gr\Gamma. 
\end{align}
Equation \eqref{eqn:jump_of_v_at_gamma} shows that the above expression is indeed well-defined, i.e.\ the two representations are equal. Note that $\Gamma$ is no longer a material interface. The mass transfer term can cause a motion of the contact line. Hence the Navier condition at $\Gamma$ reads
\begin{align*}
\inproduct{S^\pm \ndomega}{\ngamma} &= - \lambda^\pm v^\pm \cdot \ngamma\\
&= - \lambda^\pm\left( \clspeed + \frac{\dot{m}}{\rho^\pm \sin \theta}\right).
\end{align*}
We, therefore, obtain the jump condition
\begin{align}
\inproduct{\jump{S}\ndomega}{\ngamma} = - \jump{\lambda}\clspeed - \frac{\dot{m}}{\sin\theta}\jump{\frac{\lambda}{\rho}}
\end{align}
on $\gr\Gamma$. This finally leads to
\begin{align*}
\jump{\visc} \DSigmaDT{\theta} = \frac{\jump{\lambda}\clspeed}{2} - \frac{\partial_\tau \sigma}{2} - \jump{\visc \, \partial_\tau \left( \frac{\dot{m}}{\rho} \right)} \\
- \frac{\dot{m}}{\sin \theta} \left( \kappa \cos \theta \jump{\frac{\visc}{\rho}} + \frac{\cos\theta}{2} \jump{\frac{1}{\rho}} \dot{m} - \frac{1}{2}\jump{\frac{\lambda}{\rho}} \right).
\end{align*}
The claim follows from 
\[ \jump{\visc \, \partial_\tau \left( \frac{\dot{m}}{\rho} \right)} = \jump{\visc/\rho} \partial_\tau \dot{m} + \dot{m} \jump{\visc \partial_\tau (1/\rho)}.\qedhere  \]
\end{proof}

\section{Conclusion}
The kinematic evolution equation for the dynamic contact angle \eqref{eqn:main_result} expresses the relation between the rate of change of the contact angle and the structure of the transporting interfacial velocity field $\vsigma$ defined on the moving interface with boundary. It holds whenever $\vsigma$ is continuously differentiable and can, therefore, be applied to study \textit{potential} regular solutions for various models of dynamic wetting.\newline
\newline
We used \eqref{eqn:main_result} to derive the contact angle evolution equation \eqref{eqn:evolution_equation_standard_model} which describes the time evolution of the contact angle for \emph{potential} regular solutions of the two-phase incompressible Navier-Stokes equations with a Navier slip boundary condition. Together with the usual modeling for the dynamic contact angle respecting the condition \eqref{eqn:advancing_receding_condition}, the contact angle evolution turns out to be unphysical if the slip length is positive and finite. Hence a (weak) singularity is present at the contact line even if the contact angle is allowed to vary. Note that the presence of a singularity might cause challenges for numerical simulations as pointed out in \cite{Sui.2014},\cite{Sprittles.2011}. It might, therefore, be necessary to develop numerical methods which incorporate some a priori knowledge about the singularity at the contact line, see for example \cite{Sprittles.2011b}.\newline
\newline
We prove that, for regular solutions, the contact angle stays constant over time if the slip length is infinite at the contact line. In this case, the pressure is regular at the contact line as pointed out for example in \cite{Shikhmurzaev.2006}.\newline
\newline
We also studied the contact angle evolution for regular solutions to more general models in Section~\ref{section:generalizations}. Interestingly, a surface tension gradient at the contact line may give rise to physically reasonable regular solutions, see \eqref{eqn:ca_evolution_marangoni}. We also analyzed the case of slip at the fluid-fluid interface, see \eqref{eqn:ca_evolution_with_if_slip}. Here we find qualitatively similar results if
\[ \frac{\jump{\lambda}}{\jump{\visc}}_{|\Gamma} \geq 0. \]
Finally, we applied the kinematic evolution equation to a class of models including phase transfer across the fluid-fluid interface, see \eqref{eqn:contact_angle_evolution_mass_transfer_1}. In this case the interface moves with its own velocity and, hence, additional terms, depending on the mass transfer rate $\dot{m}$, appear in the contact angle evolution equation.\newline
\newline
Moreover, the kinematic evolution equation \eqref{eqn:main_result} can be used as a reference to validate numerical methods. We will address the advective transport of the contact angle by a prescribed velocity field in a forthcoming paper.

\paragraph{Acknowledgements:} We kindly acknowledge the financial support by the German Research Foundation (DFG) within the Collaborative Research Center 1194 “Interaction of Transport and Wetting Processes”, Project B01.

\clearpage
\appendix
\section{Appendix}
\begin{lemma}[Separated local parametrization]
\label{lemma:local_parametrization_old_appendix}
Let $\{\Sigma(t)\}_{t \in I}$ be a $\mathcal{C}^{1,2}$-family of moving hypersurfaces and $(t_0,x_0)$ be an inner point of $\move=\gr\Sigma$. Then there exists an open neighborhood $U \subseteq \RR^4$ of $(t_0,x_0)$, $\delta, \varepsilon > 0$ and a $\mathcal{C}^1$-parametrization
\begin{align*}
\phi: \, \underbrace{(t_0 - \delta, t_0 + \delta)}_{=:I_\delta(t_0)} \times B^2_\varepsilon(0) \, \rightarrow \, \move \cap U
\end{align*}
of $\move$ such that $\phi(t_0,0) = (t_0,x_0)$ and
\begin{align*}
\phi(t,\cdot): \, B^2_\varepsilon(0) \, \rightarrow \, \{t\} \times \Sigma(t)
\end{align*}
is a $\mathcal{C}^2$-parametrization of $\Sigma(t)$. In particular
\begin{align*}
\phi(t,u) = (t,\hat{\phi}(t,u)), 
\end{align*}
with a $\mathcal{C}^1$-function $\hat{\phi}$.
\end{lemma}
\begin{proof}
By definition of a $\mathcal{C}^{1,2}$-family of moving hypersurfaces, there is $\eta > 0$ and an open neighborhood of $(t_0,x_0) \in \RR^4$ and a local $\mathcal{C}^1$-parametrization
\begin{align}
\psi = (\psi_t,\psi_x): \RR^3 \subset B^3_\eta(0) \rightarrow \gr \Sigma \cap U
\end{align}
such that $\psi(0) = (t_0,x_0)$ and $\psi_x(u) \in \Sigma(\psi_t(u))$ for all $u \in B^3_\eta(0)$. The goal is to find a coordinate transformation 
\[ T: I_\delta(t_0) \times B^2_\varepsilon(0) \rightarrow B^3_\eta(0) \] 
such that
\begin{align*} 
\psi_t(T(s,y_1,y_2)) = s. 
\end{align*}
Since $\psi$ is injective, there is $i \in \{1,2,3\}$ such that
\begin{align}
\label{eqn:invertibility_condition_old}
(\partial_{u_i} \psi_t)(0,0,0) \neq 0, 
\end{align}
where we may assume $i=1$. We now choose a special function $T$ of the form
\[ T(s,y_1,y_2) = (\varphi(s,y_1,y_2),y_1,y_2) \]
and look for a function $\varphi$ satisfying
\begin{align*} 
&\psi_t(\varphi(s,y_1,y_2),y_1,y_2) = s\\
\Leftrightarrow \quad &0 = f(s,y_1,y_2;\varphi(s,y_1,y_2)).
\end{align*}
The $\mathcal{C}^1$-function $f(s,y_1,y_2;\varphi):=\psi_t(\varphi,y_1,y_2)-s$ satisfies $f(t_0,0,0;0)=0$ and \eqref{eqn:invertibility_condition_old} implies
\[ \partial_\varphi f(t_0,0,0;0) \neq 0. \]
Now the claim follows by the Implicit Function Theorem.
\end{proof}
Note that exactly the same procedure yields a $\mathcal{C}^1$-parametrization of the submanifold $\gr\Gamma$ of the form
\[ \phi: \, (t_0 - \delta, t_0 + \delta) \times (-\varepsilon,\varepsilon) \, \rightarrow \, \gr\Gamma \cap U \]
such that $\phi(t_0,\cdot)$ is a $\mathcal{C}^1$-parametrization of $\Gamma(t_0)$. As a consequence of that, we can give an explicit characterization of the tangent spaces of $\gr\Sigma$ and $\gr\Gamma$.

\begin{lemma}[Tangent spaces]
\label{lemma:tangent_spaces_appendix}
The tangent space of $\gr\Sigma$ at the point $(t,x)$ is given by
\begin{align*}
T_{\gr\Sigma}(t,x) = \{ \lambda \, (1,\normalspeed \nsigma(t,x)) + (0,\tau):\\
\lambda \in \RR, \, \tau \in T_{\Sigma(t)}(x) \}.
\end{align*}
Likewise, the tangent space of $\gr\Gamma$ at the point $(t,x)$ is given by
\begin{align*}
T_{\gr\Gamma}(t,x) = \{ \lambda \, (1,\clspeed \ngamma(t,x)) + (0,\tau):\\
\lambda \in \RR, \, \tau \in T_{\Gamma(t)}(t,x) \}.
\end{align*}
\end{lemma}
\begin{proof}
We make use of the parametrization constructed in Lemma~\ref{lemma:local_parametrization_old_appendix}.\newline
\newline
For $(t_0,x_0) \in \gr\Sigma$ choose a $\mathcal{C}^1$-parametrization
 \begin{align*} 
 \phi: \, (t_0 - \delta, t_0 + \delta) \times B^2_\varepsilon(0) \, \rightarrow \, \move \cap U,\\
 \phi(t,u_1,u_2) = (t,\hat{\phi}(t,u_1,u_2)) 
 \end{align*}
 such that $\phi(t_0,0,0) = (t_0,x_0)$. A basis for the tangent space $\tangentspace{\gr\Sigma}(t_0,x_0)$ is then given by
 \begin{align*} 
 \{ \partial_t\phi(t_0,0,0), \partial_{u_1}\phi(t_0,0,0), \partial_{u_2}\phi(t_0,0,0) \} \\
 = \{(1,\partial_t \hat{\phi}(t_0,0,0)), (0,\partial_{u_1} \hat{\phi}(t_0,0,0)),\\
 (0,\partial_{u_2} \hat{\phi}(t_0,0,0)) \},
 \end{align*}
 where the vectors 
 \[ v_1:=\partial_{u_1} \hat{\phi}(t_0,0,0) \quad \text{and} \quad v_2:=~\partial_{u_2} \hat{\phi}(t_0,0,0) \]
 constitute a basis of $T_{\Sigma(t_0)}(x_0)$. By definition of the normal velocity $\normalspeed$, we have
 \begin{align*}
 \partial_t \hat{\phi}(t_0,0,0) = &\inproduct{\partial_t \hat{\phi}(t_0,0,0)}{\nsigma(t_0,x_0)}\nsigma(t_0,x_0)\\
 &+ \psigma \, \partial_t \hat{\phi}(t_0,0,0) \\
 = &\normalspeed(t_0,x_0) \nsigma(t_0,x_0) + \psigma \, \partial_t \hat{\phi}(t_0,0,0).
 \end{align*}
 Since the second term can be expressed in terms of $v_1$ and $v_2$, we obtain a basis of the desired form
 \begin{align*} 
 \{(1,\normalspeed(t_0,x_0)\nsigma(t_0,x_0)), (0,\partial_{u_1} \hat{\phi}(t_0,0,0)),\\
 (0,\partial_{u_2} \hat{\phi}(t_0,0,0)) \}. 
 \end{align*}
 For $(t_0,x_0) \in \gr\Gamma$ choose a $\mathcal{C}^1$-parametrization
 \begin{align*} 
 \phi: \, (t_0 - \delta, t_0 + \delta) \times (-\varepsilon,\varepsilon) \, \rightarrow \, \gr\Gamma \cap U,\\
 \phi(t,u) = (t,\hat{\phi}(t,u)), 
 \end{align*}
 such that $\phi(t_0,0) = (t_0,x_0)$. The same procedure as above shows that the set
 \[ \{(1,\clspeed\ngamma(t_0,x_0)), (0,\partial_u \hat{\phi}(t_0,0)  \} \]
 is a basis of the Tangent space $\tangentspace{\gr\Gamma}(t_0,x_0)$, where $\partial_u \hat{\phi}(t_0,0)$ is a basis of $\tangentspace{\Gamma(t_0)}(x_0)$.\qedhere

\end{proof}

\begin{lemma}[Signed distance function]
Let $\{\Sigma(t)\}_{t \in I}$ be a $\mathcal{C}^{1,2}$-family of moving hypersurfaces and $(t_0,x_0)$ be an inner point of $\move=\gr\Sigma$. Then there exists an open neighborhood $U \subset \RR^4$ of $(t_0,x_0)$ and $\varepsilon > 0$ such that the map
\begin{align*}
\tube: (\move \cap U) \times (-\varepsilon, \varepsilon) \rightarrow \RR^4,\\ 
\tube(t,x,h) := (t, x + h \nsigma(t,x))\nonumber
\end{align*}
is a diffeomorphism onto its image
\begin{align*}
\mathcal{N}^\varepsilon := \tube((\move \cap U) \times (-\varepsilon \times \varepsilon)) \subset \RR^4 ,
\end{align*}
i.e.\ $\tube$ is invertible there and both $\tube$ and $\tube^{-1}$ are $\mathcal{C}^1$. The inverse function has the form
\begin{align*}
\tube^{-1}(t,x) = (\pi(t,x),d(t,x))
\end{align*}
with $\mathcal{C}^1$-functions $\pi$ and $d$ on $\mathcal{N}^\varepsilon$.
\end{lemma}
\begin{proof}
According to Lemma~\ref{lemma:local_parametrization_old_appendix}, we can choose a local $\mathcal{C}^1$-parametrization $\phi$ of $\move$ of the form (with $\delta, \varepsilon > 0, \, U_0 \subset \RR^4$ open)
\begin{align*} 
\phi: \underbrace{(t_0-\delta, t_0 + \delta)}_{=:I_\delta(t_0)} \times B^2_\varepsilon(0) \rightarrow \move \cap U_0,
\end{align*}
where $\phi$ has the following form
\begin{align*}
\phi(t,u) = (t,\hat{\phi}(t,u)), \quad \phi^{-1}(t,x) = (t,u(t,x)). 
\end{align*}
Then $\tube$ can be expressed as
\begin{align*}
\tube(t,x,h) = \tube^0(\phi^{-1}(t,x),h)
\end{align*}
with
\begin{align*}
\tube^0(t,u,h) := (t,\hat{\phi}(t,u) + h \nsigma(t,\hat{\phi}(t,u))). 
\end{align*}
The function
\begin{align*}
\tube^0: I_\delta(t_0) \times B^2_\varepsilon(0) \times \RR \rightarrow \RR^4 
\end{align*}
is continuously differentiable (since $\hat{\phi} \in \mathcal{C}^1(I_\delta(t_0) \times B^2_\varepsilon(0))$ and $\nsigma \in \mathcal{C}^1(\move)$) and the Jacobian of $X^0$ has the form
\begin{align*}
(D\, \tube^0)(t,u,h) = 
\left(\begin{matrix}
1 & 0 & 0 & 0\\%
\ast & \multicolumn{3}{c}{}\\%
\ast & \multicolumn{3}{c}{D\, \tube^0_t(u,h)}\\%
\ast & \multicolumn{3}{c}{}\\%
\end{matrix}\right),
\end{align*}
where $\tube^0_t$ corresponds to $\tube^0$ at fixed $t$, i.e.
\[ \tube^0_t(u,h) := \hat{\phi}(t,u) + h \nsigma(t,\hat{\phi}(t,u)). \]
Obviously, $D\tube^0$ is invertible at $(t,u,h)$ if and only if $D\tube^0_t$ is invertible at $(u,h)$. The invertibility of $D\tube^0_{t_0}$ at the point $(0,0)$ is a well-known result from the theory of $\mathcal{C}^2$-hypersurfaces. In particular, one can show by the Banach contraction principle that $D\tube^0_t$ is invertible on $B_{\frac{\varepsilon}{2}}(0) \times (-\varepsilon(t),\varepsilon(t))$ if (see \cite{Pruss.2016}, chapter 2.3 for details)
\begin{align}
\label{eqn:invertibility_condition}
\varepsilon(t) \norm{\nablasigma \nsigma(t,\hat{\phi}(t,\cdot)}_{\mathcal{C}\left(\overline{B_{\frac{\varepsilon}{2}}(0)}\right)} < 1. 
\end{align}
Since $\nsigma \in \mathcal{C}^1(\move)$, we can choose for every compact subset $I \subset I_\delta(t_0)$ an $\hat{\varepsilon} > 0$ such that \eqref{eqn:invertibility_condition} holds for all $t \in I$ with $\varepsilon(t):=\hat{\varepsilon}$. In particular, $D\tube^0$ is invertible at the point $(t_0,0,0)$. Now it follows from the Implicit Function Theorem that there are open neighborhoods $V$ of $(t_0,0,0)$ and $U\subseteq U_0$ of $(t_0,x_0) \in \RR^4$ such that $\tube^0: V \rightarrow U$ is a bijection and both $\tube^0$ and $(\tube^0)^{-1}$ are $\mathcal{C}^1$. Since the parametrization $\phi: I_\delta(t_0) \times B^2_\varepsilon(0) \rightarrow \move \cap U_0$ is a diffeomorphism between manifolds, the claim for $\tube$ follows from the properties of $\tube^0$.
\end{proof}

\bibliography{ms.bbl}

\end{document}